\documentclass[pdflatex,sn-vancouver-ay]{sn-jnl}
\pdfoutput=1

\usepackage{graphicx}%
\usepackage{multirow}%
\usepackage{amsmath,amssymb,amsfonts}%
\usepackage{amsthm}%
\usepackage{mathrsfs}%
\usepackage[title]{appendix}%
\usepackage{xcolor}%
\usepackage{textcomp}%
\usepackage{manyfoot}%
\usepackage{booktabs}%
\usepackage{algorithm}%
\usepackage{algorithmicx}%
\usepackage{algpseudocode}%
\usepackage{listings}%

\usepackage{makecell}
\usepackage{tabularx}
\usepackage{here}
\usepackage{booktabs}
\usepackage{enumitem}

\DeclareFontFamily{U}{mathx}{\hyphenchar\font45}
\DeclareFontShape{U}{mathx}{m}{n}{
      <5> <6> <7> <8> <9> <10>
      <10.95> <12> <14.4> <17.28> <20.74> <24.88>
      mathx10
      }{}
\DeclareSymbolFont{mathx}{U}{mathx}{m}{n}
\DeclareFontSubstitution{U}{mathx}{m}{n}
\DeclareMathAccent{\widecheck}{0}{mathx}{"71}
\DeclareMathAccent{\wideparen}{0}{mathx}{"75}

\newcommand{\calJ}{{\mathcal J}}    
\newcommand{\calK}{{\mathcal K}}    
    
\newcommand{\calM}{{\mathcal M}}

\newcommand{\calU}{{\mathcal U}}

\newcommand{\calX}{{\mathcal X}}

\newcommand{\bbE}{{\mathbb E}}

\newcommand{\bbP}{{\mathbb P}}    
    
\newcommand{\bbR}{{\mathbb R}}

\newcommand{\bbU}{{\mathbb U}}

\newcommand{\bbZ}{{\mathbb Z}}

\newcommand{\bmx}{{\mathbf x}}
\newcommand{\bmy}{{\mathbf y}}

\newcommand{\tilA}{{A}}

\DeclareMathOperator*{\argmin}{arg\,min}

\newcommand{\diag}{\mathop{\rm diag}\nolimits}

\newcommand{\rmc}{{\rm c}}

\newcommand{\rmx}{{\rm x}}

\newcommand{\rmu}{{\rm u}}
\newcommand{\rmf}{{\rm f}}

\newcommand{\Ee}{{\rm e}}   

\newcommand{\bbra}[1]{\ensuremath{[\![#1]\!]} }  

\newcommand{\ft}{T}
\newcommand{\thr}{d}

\newcommand{\axv}[1]{{\color{black}{#1}}} 
 
\newcommand{\blue}[1]{{\color{black}{#1}}}

\newcommand{\memo}[1]{{\color{black}{#1}}}

\newcommand{\wtilde}[1]{{\widetilde{#1}}}
\newcommand{\wcheck}[1]{{\widecheck{#1}}}

\newcommand{\usigma}{\underline{\sigma}} 

\newcommand{\regret}{{\rm Regret}}

\newcommand{\dime}{n}
\newcommand{\surr}{F}

\newcommand{\bone}{\boldsymbol{1}}

\newcommand{\term}{D}
\newcommand{\coe}{C}

\makeatletter
\newcommand{\vast}{\bBigg@{4}}
\newcommand{\Vast}{\bBigg@{5}}
\makeatother

\theoremstyle{thmstyleone}%
\newtheorem{theorem}{Theorem}
\newtheorem{proposition}[theorem]{Proposition}%
\newtheorem{lemma}[theorem]{Lemma}%

\theoremstyle{thmstyletwo}%
\newtheorem{remark}{Remark}%

\theoremstyle{thmstylethree}%
\newtheorem{definition}{Definition}%
\newtheorem{assumption}[theorem]{Assumption}

\begin{document}

\title[Online Control of Linear Systems under Unbounded Noise]{\centering Online Control of Linear Systems \\under Unbounded Noise}


\author*[1]{\fnm{Kaito} \sur{Ito}}\email{kaito@g.ecc.u-tokyo.ac.jp}

\author[2]{\fnm{Taira} \sur{Tsuchiya}}\email{tsuchiya@mist.i.u-tokyo.ac.jp}

\affil*[1]{\orgdiv{Department of Information Physics and Computing}, \orgname{The University of Tokyo}, \orgaddress{\street{7-3-1 Hongo, Bunkyo-ku}, \city{Tokyo}, \postcode{1138656}, \country{Japan}}}

\affil[2]{\orgdiv{Department of Mathematical Informatics}, \orgname{The University of Tokyo}, \orgaddress{\street{7-3-1 Hongo, Bunkyo-ku}, \city{Tokyo}, \postcode{1138656}, \country{Japan}}}





\abstract{This paper investigates the problem of controlling a linear system under possibly unbounded stochastic noise with unknown convex cost functions, known as an online control problem.
    In contrast to the existing work, which assumes the boundedness of noise, we show that an $ \wtilde{O}(\sqrt{T}) $ high-probability regret can be achieved under unbounded noise, where $ T $ denotes the time horizon. Notably, the noise is only required to have finite fourth moment.
    Moreover, when the costs are strongly convex and the noise is sub-Gaussian, we establish an $ O({\rm poly} (\log T)) $ regret bound.}




\maketitle

\allowdisplaybreaks[0]

\section{Introduction}\label{sec1}

\begin{table*}[t]
    \caption{Regret comparisons for online control of known linear systems with unknown costs. Adv. stoc., and h.p. are the abbreviations of adversarial disturbance, stochastic noise, and high probability, respectively.
    }
    \label{table:regret_comparison}
    \centering
    \small
    \begin{tabular}{m{3.2cm}lll}
      \toprule
      Reference & Disturbance  &  Cost functions & Regret bound \\
      \midrule
      \citet{Agarwal2019} & adv., bounded    &convex & $\displaystyle  O ( \sqrt{T} (\log T)^{3.5})$   \\
      \citet{Cohen2018} & Gaussian  & quadratic & \makecell{ $\displaystyle  O(  \sqrt{T} ) $ \\ pseudo-regret bound  } \\
      \citet{Agarwal2019log} & stoc., bounded  & strongly convex & \makecell{$\displaystyle  O(  (\log T)^7 )$ \\pseudo-regret bound }  \\
      \midrule
      \makecell{\textbf{This work}\\ (Theorem~\ref{thm:regret})} & \makecell{stoc., unbounded, \\ finite fourth moment}   &convex & \makecell{$\displaystyle O( \sqrt{T} \memo{(\log T)^{3}})$ \\that holds with h.p. } \\
      \makecell{\textbf{This work} \\ (Theorem~\ref{thm:log_regret})} & \makecell{stoc., unbounded,\\ sub-Gaussian}   & \makecell{strongly convex, \\ smooth} & \makecell{ $\displaystyle  O( \memo{(\log T)^{8}} )$\\ that holds with h.p.  }\\
      \bottomrule 
    \end{tabular}
  \end{table*}

This paper addresses the problem of controlling a linear dynamical system under possibly unbounded stochastic noise while minimizing a given control performance index.
In particular, we consider the situation where the cost function at each time step is revealed to a controller only after the control input for that time step has been determined.
This problem has been investigated in the framework of regret minimization under the name of online control~\citep{Abbasi2014,Cohen2018,Agarwal2019}, which can be seen as an extension of online learning~\citep{Hazan2016intro,Orabona2023modern} to dynamical systems.

In the literature of online control, two scenarios have been considered: disturbances to a dynamical system are adversarial (arbitrary) or stochastic. 
Online control with adversarial disturbances takes worst-case perturbations into account. However, the adversarial disturbance model can be overly pessimistic when the statistical properties of disturbances are available. Stochastic disturbances are often referred to as noise in control theory and have been used to model probabilistic uncertainty in dynamical environments.
Consequently, it is important to develop online control algorithms for both of adversarial disturbances and stochastic noise.

In real-world applications, various types of noise can arise.
In particular, unbounded noise is troublesome for safety-critical systems.
For example, when modeling the fluctuation of wind and solar power as stochastic noise to power systems, it is crucial to consider the impact of their extreme outliers on the power systems~\citep{Kashima2019,Ito2021,Schafer2018}.
Bounded noise cannot capture such uncertainty, and a heavy-tailed distribution would be more suitable for modeling such a scenario.
Despite the importance of unbounded noise, most existing studies on online control can be applied only to bounded noise.
To our knowledge, the sole exception is the work of \citet{Cohen2018}, which only focuses on Gaussian noise and pseudo-regret.
Therefore, the applicability of existing online control algorithms in real-world scenarios has not been sufficiently established.

\subsection{Contributions}
In order to overcome this limitation, this paper investigates online control under unbounded noise.  In particular, we show that a sublinear regret bound can be achieved if the noise has a finite fourth moment.
Moreover, by additionally assuming that the noise is sub-Gaussian and the cost functions are strongly convex, we can attain a logarithmic regret bound.
This study is the first one that characterizes the regret bound of online control in terms of the statistical property of the noise.
The main contributions of this paper are as follows (see Table~\ref{table:regret_comparison} for a comparison with existing regret bounds):
\begin{enumerate}
    \item For general convex costs,
    we prove that the online control algorithm in \citet{Agarwal2019} with a properly modified learning rate achieves a high-probability regret bound of \memo{$ O(\sqrt{T}(\log T)^{3}) $} for time horizon $T$.
    This regret guarantee holds even when the noise is unbounded, and the noise is only required to have a finite fourth moment.
    \memo{As a byproduct of our careful analysis for bounding the gradient of the surrogate (or idealized) cost function introduced in \citet{Agarwal2019}}, our bound slightly improves the existing upper bound of $ O(\sqrt{T} (\log T)^{3.5}) $ for bounded noise in \citet{Agarwal2019}.
    \item For sub-Gaussian noise and strongly convex and smooth cost functions, we establish a high-probability regret upper bound of \memo{$ O((\log T)^{8}) $}.
    This poly-logarithmic regret is made possible by leveraging the exponentially fast decay of the tail of a noise distribution.
    Note that our bound is not comparable with the regret bound of $O((\log T)^7)$ in \citet{Agarwal2019log}, in which they focus on the {\em pseudo}-regret. In contrast to the pseudo-regret, bounding the regret in online control is more challenging because the strong convexity of the surrogate costs to be optimized only holds in expectation as pointed out in \citet{Simchowitz2020improper}.
\end{enumerate}

\subsection{Related Work}
Regret analysis for learning control policies has been investigated in several settings.
When state and input matrices, which characterize a linear system, are unknown but quadratic cost functions are known to the controller, it is known that the $ O(\sqrt{T}) $ regret is tight in the presence of Gaussian noise~\citep{Simchowitz2020}. In the above setting with a known input matrix, an $ O((\log T)^2) $ regret bound can be achieved~\citep{Cassel2020}.
\citet{Foster2020} considered the setting where only bounded disturbances are unknown and established an $ O((\log T)^3) $ regret bound.
\citet{Goel2023} also studied this setting and derived a regret optimal control policy, which minimizes worst-case regret. Although disturbances are not assumed to be bounded, the regret optimal control problem imposes a penalty on the norm of the disturbances so that they do not become so large.
\citet{Hazan2020nonstochastic} investigated the online control under adversarial disturbances where both system parameters and cost functions are unknown and provided an $ O(T^{2/3}) $ regret bound.
Even when the state of a system can be partially observed, similar regret bounds \axv{were} obtained by~\citet{Simchowitz2020improper}.

The most closely related studies to the present paper are \citet{Agarwal2019,Agarwal2019log}.
The former proposed to use the disturbance-action policy, which allows us to reduce the online control problem to online convex optimization (OCO) with memory~\citep{Anava2015}.
The latter work revealed that if cost functions are strongly convex, then a logarithmic pseudo-regret can be achieved.
Note that their regret analysis heavily relies on the boundedness of noise.

There has been very limited work on online control where cost functions are unknown and noise is unbounded.
To our knowledge, \citet{Cohen2018} is the only study that addresses unbounded noise in the context of online control, focusing on a linear system with Gaussian noise and unknown quadratic costs.
However, they only studied a weaker performance index of the pseudo-regret, and moreover, since their algorithm does not use the curvature of the costs, only $ O(\sqrt{T}) $ pseudo-regret bound can be achieved. 
In contrast, our approach can be applied to the same setting to achieve the logarithmic regret.
Furthermore, our approach is applicable to general convex costs whereas their approach focuses on quadratic costs.

{\it Notation:}
Let $ \bbR $ denote the set of real numbers and $ \bbZ_{>0} $ (resp.~$ \bbZ_{\ge 0} $) denote the set of positive (resp.~nonnegative) integers. The set of integers $ \{k,k+1,\ldots,\ell\}  $ for $ k\le \ell $ is denoted by $ \bbra{k,\ell} $.
A sequence of matricies $ \{A_k,A_{k+1},\ldots,A_\ell \} $ for $ k\le \ell $ is denoted by $ A_{k:\ell} $.
For a symmetric matrix $ A $, we use $ A \succ 0 $ (resp.~$ A \succeq 0 $) to mean that $ A $ is positive definite (resp.~semidefinite).
The spectral norm and the Frobenius norm are denoted by $ \| \cdot \| $ and $ \|\cdot \|_{\rm F} $, respectively. Note that for a vector, the spectral norm coincides with the Euclidean norm.
For $ M = \{M^{[0]},\ldots,M^{[\ell-1]} \} \in (\bbR^{n\times m})^\ell $, we denote its Frobenius norm by $ \|M\|_{\rm F} = \| [ M^{[0]},\ldots, M^{[\ell-1]} ] \|_{\rm F}$.
For $ a \in \bbR $, the smallest integer greater than or equal to $ a $ is denoted by $ \lceil a \rceil $.
The natural logarithm of $ a > 0 $ is written as $ \log a $.
To simplify the notation, we sometimes write $f = O(g)$ as $f \lesssim g$.
The indicator of a condition $ S $ is denoted by $ \bone_{S} $, i.e., $ \bone_{S} = 1 $ if the condition $ S $ holds, and $ 0 $, otherwise.
The gradient and Hessian matrix of a function $ f $ with respect to the variable $ x $ are denoted by $ \nabla_x f $ and $ \nabla_x^2 f $, respectively.
\blue{For $ \xi > 0 $ and $ \calX \subseteq \bbR^n $, a twice differentiable function $ f : \calX \rightarrow \bbR $ is said to be $ \xi $-strongly convex if $ \nabla_x^2 f (x) \succeq \xi I $ for any $ x\in \calX $,
where $ I $ denotes the identity matrix of appropriate dimension.}
All random variables are defined on a common probability space with measure $ \bbP $, and the expectation with respect to $ \bbP $ is denoted by $ \bbE $.


\section{Problem Formulation}
In this paper, we consider the following linear dynamical system:
\begin{equation}\label{eq:system}
  x_{t+1} = A x_t + B u_t + w_t  , ~~ t \in \bbZ_{\ge 0},
\end{equation}
where $ \{x_t\}_{t\ge 0} $ is an $ \bbR^{n_\rmx} $-valued state process, $ \{u_t\}_{t\ge 0} $ is an $ \bbR^{n_\rmu} $-valued control process, and $ \{w_t\}_{t\ge 0} $ is an $ \bbR^{n_\rmx} $-valued stochastic noise process.
Moreover, we assume that the system matrices $ A\in \bbR^{n_\rmx \times n_\rmx} $ and $ B\in \bbR^{n_\rmx\times n_\rmu} $ are known and let $  \kappa_B := \max \left\{\|B\|, 1 \right\} $. For simplicity, we further assume $ x_0 = 0 $.\footnote{Noting that the influence of initial states on stable systems decays exponentially over time, we can extend our results to nonzero initial states.}
\blue{At each time $ t = 0,\ldots,T-1 $, the controller (or learner) chooses a control input $ u_t $ and incurs a cost $ c_t (x_t,u_t) $, where $ c_t : \bbR^{n_\rmx} \times \bbR^{n_\rmu} \rightarrow \bbR $ is the cost function.
In online control problems, we can use only the current and past states $ \{x_s\}_{s=0}^t $ and the past cost functions $ \{c_s\}_{s=0}^{t-1} $ when determining the control $ u_t $. Then, the controller aims to generate an adaptive sequence of control $ \{u_t\} $ so that the cumulative cost $ J_T \left(\{u_t\} \right) := \sum_{t=0}^{T-1} c_t (x_t,u_t) $ over a finite horizon $ T \in \bbZ_{>0} $ is minimized.}
The performance index for a given online control algorithm is defined by the regret, which is given by
\begin{align}
    \regret_T := J_T (\{u_t\}) - \min_{\{u_t^*\} \in \calU} J_T(\{u_t^*\}) ,\label{eq:regret}
\end{align}
where $ \{u_t\} $ is determined by the online control algorithm, and $ \calU $ is an admissible set of control processes (see \eqref{eq:admissible_u} for the definition).
For the second term of \eqref{eq:regret}, all the cost functions $ \{c_t\}_{t=0}^{T-1} $ and the realization of the noise sequence $ \{w_t\}_{t=0}^{\ft-1} $ are assumed to be available for the optimization of $ \{u_t^*\} $, and we assume that there exists an optimal solution $ \{u_t^*\} $ throughout the paper.
For notational simplicity, the dependency of the regret on $ \{c_t\}_{t=0}^{T-1} $ and $ \{w_t\}_{t=0}^{\ft-1} $ is omitted.

For a specific choice of $ \calU $, the set of {\em strongly stable} control gains are often used in the literature~(e.g.,~\citealt{Cohen2018,Agarwal2019,Hazan2020nonstochastic}).
We first introduce the notion of strong stability.
\begin{definition}\label{def:stability}
    For $ \kappa \ge 1 $ and $ \gamma \in (0,1) $, a matrix $ K \in \bbR^{n_\rmu \times n_\rmx} $ is said to be $ (\kappa,\gamma) $-strongly stable for system~\eqref{eq:system} if there exist $ n_\rmx \times n_\rmx $ complex matrices $ P $ and $Q$ such that $ A - BK = QPQ^{-1} $ and
    \begin{enumerate}
        \item[(i)] $ \| P \| \le 1-\gamma $,
        \item[(ii)] $ \|K \| \le \kappa $, $ \|Q \| \le \kappa $, $ \|Q^{-1} \| \le \kappa $. 
    \end{enumerate}
    Especially when $ P $ is diagonal, $ K $ is said to be $ (\kappa,\gamma) $-diagonally strongly stable.\footnote{\memo{If $ A - BK $ is diagonalizable and its spectral radius is less than $ 1 $, $ K $ is diagonally strongly stable for some $ (\kappa,\gamma) $. Further discussion of the generality of the diagonal strong stability can be found in \citet{Agarwal2019log}}.}
    \hfill $ \diamondsuit $
\end{definition}

The parameters $ \kappa $ and $\gamma $ quantify the stability of \eqref{eq:system} under the linear state-feedback control $ u_t = -Kx_t $.
Indeed, by definition, a $ (\kappa,\gamma) $-strongly stable matrix $ K $ satisfies
\begin{align}
    \| A_K^i \| \le \kappa^2 (1-\gamma)^i , ~~ \forall i \in \bbZ_{\ge 0} , ~~ A_K := A - BK  \label{eq:strong_stable_inequality}.
\end{align}
Hence, if $ w_t \equiv 0 $, the state $ x_t $ under $ u_t = -Kx_t $ satisfies
\begin{align}
    \|x_t\| = \|A_K^t x_0 \| \le \kappa^2 (1-\gamma)^t \|x_0 \| , ~~ \forall t \in \bbZ_{\ge 0} . \label{eq:exponential}
\end{align}

Using the strong stability, we define the admissible set of control $ \calU $ as
\begin{align}
    \calU(\kappa,\gamma) &:= \bigl\{  \{u_t\}_{t=0}^{T-1} : \exists K \in \calK(\kappa,\gamma) ~ {\rm s.t.}~  u_t = -Kx_t, ~ \forall t \in \bbra{0,T-1}  \bigr\} , \label{eq:admissible_u}\\
    \calK(\kappa,\gamma) &:= \{ K \in \bbR^{n_\rmu\times n_\rmx} : K~\text{is}~(\kappa,\gamma)\text{-strongly stable}\} . \nonumber
\end{align}
When imposing the diagonal strong stability on $ K $ in the definition of $ \calU $, the resulting regret is denoted by $ \regret_T^{\diag} $.

For the regret analysis, we make the following two assumptions.
\begin{assumption}\label{ass:cost}
    The cost function $ c_t : \bbR^{n_\rmx} \times \bbR^{n_\rmu} \rightarrow \bbR $ is convex and differentiable for any $ t\in \bbra{0,T-1} $. In addition, there exists $ G_c \ge 1 $ such that for any $ x\in \bbR^{n_\rmx} $, $ u\in \bbR^{n_\rmu} $, and $ t \in \bbra{0,T-1} $, it holds that 
    \begin{equation}\label{eq:cost_gradient}
        \| \nabla_x c_t (x,u) \| \le G_c \|x\|, \ \| \nabla_u c_t (x,u) \| \le G_c \|u \| .
    \end{equation}
    \hfill $ \diamondsuit $
\end{assumption}
\blue{The above assumption is a natural generalization of the one in \citet{Agarwal2019,Agarwal2019log,Hazan2020nonstochastic}, which assume that there exists $ D > 0 $ such that for all $ t $, it holds that $ \|x_t \| \le D $, $ \|u_t \| \le D $ and
$
    \| \nabla_x c_t (x,u) \| \le G_c D, \ \| \nabla_u c_t (x,u) \| \le G_c D 
$
for any $ x,u $ satisfying $ \|x \| \le D $, $ \|u \| \le D $.
However, under unbounded noise, there is no finite $ D $ such that $ \|x_t \| \le D $, $ \|u_t\| \le D $ almost surely, and thus the unboundedness naturally leads to the \axv{global} condition~\eqref{eq:cost_gradient}. For instance, this is satisfied for quadratic costs.}

For the statistical property of the noise $ \{w_t\} $, we make the following assumption.
\begin{assumption}\label{ass:noise_4}
    For the noise process $ \{w_t\} $,
    \begin{enumerate}
        \item[(i)] there exists $ \sigma_w \ge 0 $ such that $ \bbE[ \|w_t \| ] \le \sigma_w $ for any $ t \in \bbZ_{\ge 0}$;
        \item[(ii)] there exists $ \sigma_w \ge 0 $ such that $ \bbE[ \|w_t \|^4 ] \le \sigma_w^4 $ for any $ t \in \bbZ_{\ge 0} $.
        \hfill $ \diamondsuit $
    \end{enumerate}
\end{assumption}
It is noteworthy that we do not assume that $ \{w_t\} $ is an independent and identically distributed sequence.
Note also that if the condition (ii) holds, then the condition (i) also holds for the same scale parameter $ \sigma_w $ due to H\"{o}lder's inequality.

\section{Preliminaries}
In this section, we introduce the key ingredients of online control in this paper: OCO with memory and a disturbance-action policy.
\subsection{Online Convex Optimization with Memory}\label{sec:OCO}
In this subsection, we briefly introduce OCO with memory, which has been applied to online control.
In this problem, for each time step $ t \in \bbra{0, T-1} $, a learner chooses a decision $ \bmx_t \in \calX $ from a bounded convex set $ \calX \subset \bbR^p $.
The learner then observes a loss function $ \ell_t : \calX^{H+1} \to \bbR$ and suffers a loss of $\ell_t(\bmx_{t-H:t})$ (see \citealt{Anava2015} for detailed backgrounds).
We assume that for any $ t\in \bbra{0,T-1} $, the loss function $ \ell_t  $ is coordinate-wise Lipschitz continuous, that is, there exists a constant $ L_{\rmc,\ell} > 0 $ such that for any $ i \in \bbra{0,H} $, $ \bmx_1,\ldots,\bmx_H, \widecheck{\bmx}_i\in \calX  $, and $ t\in \bbra{0,T-1} $,
\begin{align}
    \left| \ell_t (\bmx_{0:i-1},\bmx_i, \bmx_{i+1:H}) - \ell_t (\bmx_{0:i-1},\widecheck{\bmx}_i, \bmx_{i+1:H})  \right| \le L_{\rmc,\ell} \|\bmx_i - \widecheck{\bmx}_i \| . \label{eq:lipshitz}
\end{align}
For $ \wtilde{\ell}_t (\bmx) := \ell_t (\bmx,\ldots,\bmx)$, let $ G_{\ell} := \sup_{t\in \bbra{0,T-1},\bmx\in \calX} \| \nabla \wtilde{\ell}_t (\bmx) \| $ and $  D_\ell := \sup_{\bmx_1,\bmx_2 \in \calX} \| \bmx_1 - \bmx_2 \| $.
Then, the learner wishes to minimize the cumulative cost $ \sum_{t=0}^{T-1} \ell_t (\bmx_{t-H:t}) $.

A standard approach to OCO with memory is online gradient descent (OGD):
\begin{align}
    \bmx_{t+1} &= \Pi_{\calX} (\bmx_t - \eta_t \nabla \wtilde{\ell}_t (\bmx_t) ) := \argmin_{\bmy\in \calX} \| \bmx_t - \eta_t \nabla \wtilde{\ell}_t (\bmx_t) - \bmy \| , \label{eq:ogd_static}
\end{align}
where $ \eta_t $ is the learning rate.
OGD with learning rate $ \eta_t \equiv D_\ell /(G_{\ell} \sqrt{T}) $ attains sublinear regret as follows~\citep[Theorem~4.6]{Agarwal2019}:
    \begin{align}
        &\sum_{t=0}^{T-1} \ell_t (\bmx_{t-H:t}) - \min_{\bmx\in \calX} \sum_{t=0}^{T-1} \wtilde{\ell}_t (\bmx) \le 3D_\ell \sqrt{G_{\ell} (G_{\ell} + L_{\rmc,\ell} H^2) T} . \label{eq:bound_memory_convention}
    \end{align}
It is worth noting that when the gradient $ \nabla \wtilde{\ell}_t $ or the Lipschitz constant $ L_{\rmc,\ell} $ is unbounded, the above bound is meaningless, and this becomes a major challenge when dealing with unbounded noise as we will see later.

\subsection{Disturbance-Action Policies}
In this work, we employ disturbance-action policies, which allow us to relate the online control problem to OCO with memory.
For a fixed state-feedback gain $ K \in \calK (\kappa, \gamma) $ and parameters $ M_t = \{M_t^{[0]},\ldots,M_t^{[H-1]} \}, \ M_t^{[i]}\in \bbR^{n_\rmu\times n_\rmx} $, $ H \in \bbZ_{>0} $, the disturbance-action policy $ \pi_t (K,M_t) $ at time $ t $ chooses the control input as
\begin{align}\label{eq:disturbance-action}
    u_t = - K x_t + \sum_{i=1}^H M_t^{[i-1]} w_{t-i} ,
\end{align}
where we set $ w_t := 0 $ for $ t < 0 $.
One of the advantages of the disturbance-action policies can be seen by the following proposition.
\begin{proposition}[{\citealt[Lemma~4.3]{Agarwal2019}}]\label{prop:transfer}
    For any $ t,h,i \in \bbZ_{\ge 0} $ such that $ h\le t $ and $ i\le H + h $, define the disturbance-state transfer matrix by
    \begin{align}
        \Psi_{t,i}^{K,h} (M_{t-h:t}) := \tilA_K^i  \bone_{i\le h} + \sum_{j=0}^h \tilA_K^j B M_{t-j}^{[i-j-1]} \bone_{i-j \in \bbra{1,H}} .  \label{eq:psi_def}
    \end{align}
    Then, $ \{x_t\} $ following system~\eqref{eq:system} under the policy $ \{\pi_t(K,M_t) \}_{t=0}^{T-1} $ satisfies for any $ h \le t-1 $,
\begin{align}
    x_{t} = \tilA_K^{h+1} x_{t-1-h} + \sum_{i=0}^{H+h} \Psi_{t-1,i}^{K,h} (M_{t-1-h:t-1}) w_{t-1-i} . \label{eq:x_solution_general}
\end{align}
\hfill $ \diamondsuit $
\end{proposition}

From the above proposition, one can see that $ x_t $ and $ u_t $ are linear in the parameters $ M_{0:t} $. This linearity plays an important role in analyzing the regret of online control algorithms based on the disturbance-action policy.
In what follows, $ x_t $ driven by the disturbance-action policies $ \{\pi_s (K,M_s) \}_{s=0}^{t-1} $ is denoted by $ x_t^K (M_{0:t-1}) $. When $ M_t \equiv M $ and $ M_t \equiv 0 $, we simplify the notation as $ x_t^K (M) $ and $ x_t^K $, respectively. We use similar notation for $ u_t $ as $ u_t^K (M_{0:t}) $, $u_t^K (M)$, and $u_t^K $.

\subsection{Surrogate Cost}
The function $ \ell_t $ considered in Subsection~\ref{sec:OCO} corresponds to $ c_t (x_t^K(M_{0:t-1}),u_t^K (M_{0:t})) $ in the context of online control, which can be seen as a function of $ M_{0:t} $.
However, since the number of parameters $ M_{0:t} $ depends on $ t $, directly applying OGD to $ c_t $ requires the memory length $ H = \ft $, which renders the bound \eqref{eq:bound_memory_convention} no longer meaningful.
In view of this, we introduce a function $ \surr_t $ which is expected to approximate $ c_t (x_t^K(M_{0:t-1}),u_t^K (M_{0:t})) $ well as in \citet{Agarwal2019}.
\blue{The key idea is that even if we disregard the state values more than $ H + 1 $ steps before the current time by setting them to zero, the resulting error vanishes exponentially fast as $ H $ becomes large. This is essentially due to the exponentially fast mixing of linear systems under the strongly stable policies; see also \eqref{eq:exponential}.}

\begin{definition}\label{def:surrogate}
    Let the surrogate state $ y_{t}^K (M_{t-1-H:t-1}) $ and input $ v_t^K (M_{t-1-H:t}) $ be
    the state $ x_t $ starting from $ x_{t-1-H} = 0 $ driven by $ \{\pi_s (K,M_{s}) \}_{s=t-1-H}^{t-1} $ and the corresponding input \axv{$ u_t $}, respectively.
    Then, the surrogate cost $ \surr_t $ is defined as follows:
    \begin{align}
        &F_t(M_{t-1-H:t}):= c_t \left(y_t^K(M_{t-1-H:t-1}), v_t^K(M_{t-1-H:t}) \right) . \label{eq:surrogate_cost}
    \end{align}
    In addition, define $ f_t (M) := \surr_t (M,\ldots,M) $.
    \hfill $ \diamondsuit $
\end{definition}

By \citet[Definition~4.4]{Agarwal2019}, the surrogate state and input admit the following expressions:
\begin{align}
    y_{t}^K (M_{t-1-H:t-1}) &:= \sum_{i=0}^{2H} \Psi_{t-1,i}^{K,H} (M_{t-1-H:t-1}) w_{t-1-i}, \label{eq:surrogate_y}\\
    v_t^K (M_{t-1-H:t}) &:= -Ky_t^K (M_{t-1-H:t-1})  + \sum_{i=1}^H M_t^{[i-1]} w_{t-i} .\label{eq:surrogate_v}
\end{align}
Note that the surrogate cost $ \surr_t $ is convex with respect to $ M_{t-1-H:t-1} $ since $ y_t^K $ and $ v_t^K $ are linear in $ M_{t-1-H:t-1} $, and $ c_t $ is convex.

\section{Regret Analysis for General Convex Costs}\label{sec:regret}

\begin{algorithm}[tb]
    \caption{Online Control with Disturbance-Action Policy and OGD}
    \label{alg:online_control}
    \begin{algorithmic}[1]
        \State {\bfseries Input:} Parameters $ \kappa \ge 1, \gamma \in (0,1) $, learning rate $ \{\eta_t\} $, $ (\kappa,\gamma) $-strongly stable matrix $ K $, initial value $ M_0 \in \calM $
        \State Define $ H := \lceil 2\gamma^{-1} \log T   \rceil $
        \For{$t=0$ {\bfseries to} $T-1$}
        \State Inject $ u_t = -Kx_t + \sum_{i=1}^H M_{t}^{[i-1]} w_{t-i} $ to \axv{system}~\eqref{eq:system}
        \State Receive the cost function $ c_t $ and pay the cost $ c_t (x_t,u_t) $
        \State Receive the new state $ x_{t+1} $
        \State Calculate the value of noise $ w_t = x_{t+1} - Ax_t - Bu_t $
        \State Define the surrogate cost $ \surr_t $ as \eqref{eq:surrogate_cost} and $f_t (M) := \surr_t (M,\ldots,M) $
        \State Update $ M_{t+1} = \Pi_{\calM} (M_t - \eta_t \nabla f_t (M_t))  $
        \EndFor
    \end{algorithmic}
 \end{algorithm}

In this section, we present the main result for the regret analysis of our problem with general convex costs.
Define the set of admissible parameters for the disturbance-action policy~\eqref{eq:disturbance-action} as 
\begin{align}
\calM &:= \bigl\{ M = \{M^{[0]},\ldots,M^{[H-1]}\} :  \|M^{[i]} \| \le 2\kappa_B\kappa^3 (1-\gamma)^i, \ \forall i\in \bbra{0,H-1}      \bigr\} . \label{eq:admissible_M}
\end{align}

\subsection{Main Result}

The following theorem gives a regret bound for general convex costs. 
\begin{theorem}\label{thm:regret}
    Suppose that Assumptions~\ref{ass:cost} and \ref{ass:noise_4}-(ii) hold.
    Then, for any $ T \ge 3 $ and $ \delta \in (0,1] $, for Algorithm~\ref{alg:online_control} with $ \eta_t \equiv (\sqrt{\ft} (\log \ft)^3)^{-1} $, with probability at least $ 1-\delta $, it holds that
    \begin{align}
        \regret_\ft &\le \left( \frac{2\sqrt{3}G_c \coe_\delta^3}{\sqrt{\gamma}} + \frac{D^2}{2} \right) \sqrt{\ft}(\log \ft)^{3}   + \frac{\coe_\delta}{2} \sqrt{\ft} \log \ft + 6 G_c \coe_\delta^2 (\log\ft)^2, \label{eq:regret_thm}
    \end{align}
    where $ D := \frac{4\kappa_B \kappa^3 \sqrt{\dime}}{\gamma} $, $ \dime := \max\{n_\rmx, n_\rmu\} $, $ \coe_\delta := \frac{65724{\sigma}_w^{[1,4]}\dime^2 G_c^2  \kappa_B^6 \kappa^{18} }{\delta \gamma^8 (1-\gamma)^4}  $, and $ {\sigma}_w^{[1,4]} := \max\{\sigma_w, \sigma_w^4\} $.
    \hfill $ \diamondsuit $
\end{theorem}

\axv{The $ \wtilde{O} (W^2 \sqrt{\ft}) $ regret bound of \citet{Agarwal2019} depends on a bound on the norm of noise $ W $, and when the noise is unbounded ($ W \rightarrow \infty $), the regret bound diverges to infinity. 
By contrast, our bound~\eqref{eq:regret_thm} depends on the scale parameter $ \sigma_{w} $ of the noise rather than $ W $, and consequently guarantees the $ \wtilde{O} (\sqrt{\ft}) $ regret with high probability even when $ W \rightarrow \infty $.}
We emphasize that the regret bound of $ \wtilde{O} (W^2 \sqrt{\ft}) $ does not directly imply that for unbounded noise, $\regret_\ft = \wtilde{O}(\sqrt{T}) $ holds with high probability.
To see this, let us assume that $ w_t $ has a finite $ p $th moment ($ p \ge 1 $) satisfying $ \bbE[\|w_t\|^p] \le \sigma_w^p$ for any $ t \in \bbra{0,\ft} $. Then, we have (see e.g., \citealt[Section~5.1]{Van2014})
\[
    \bbP\left(\max_{t\in \bbra{0,\ft-1}} \|w_t \| \ge W \right) \le \frac{\sigma_w}{W} \ft^{1/p} , \ \forall W > 0 .
\]
By combining this with the bound $ \wtilde{O} (W^2 \sqrt{\ft}) $, we can see that with probability at least $ 1- \delta $, it holds that $ \regret_\ft = \wtilde{O}( \ft^{\frac{1}{2} + \frac{2}{p}}  / \delta^2 ) $. Under Assumption~\ref{ass:noise_4}-(ii), which corresponds to $ p = 4 $, this reduces to a linear bound $\wtilde{O} (\ft)$.
We also note that even if the $ p $th moment of $ w_t $ is finite for any $ p \ge 1 $, the above bound only guarantees $ \regret_\ft = \wtilde{O} ( \ft^{\frac{1}{2} + q} /\delta^2 ) $ for any $ q \in (0,2] $, which does not necessarily mean $ \wtilde{O} (\sqrt{\ft} / \delta^2) $. In contrast to the approach based on the regret bound for bounded noise, Theorem~\ref{thm:regret} reveals that the finiteness of the fourth moment of $ w_t $ is sufficient for the $ \wtilde{O} (\sqrt{\ft} / \delta^3) $ regret bound to hold.

The proof of Theorem~\ref{thm:regret} given below relies on the following decomposition of the regret, which is commonly used in the literature on online control based on OCO with memory~\citep{Agarwal2019,Agarwal2019log,Simchowitz2020improper}:
\begin{subequations}\label{eq:regret_decomposition}
\begin{align}
    \regret_T 
    &= \min_{M\in \calM} \sum_{t=0}^{T-1} f_t (M) -  \sum_{t=0}^{T-1} c_t (x_t^{K^*},u_t^{K^*}) \label{eq:regret_decomposition1}\\
    &\quad + \sum_{t=0}^{T-1} c_t (x_t^K (M_{0:t-1}), u_t^K(M_{0:t}))  - \sum_{t=0}^{T-1} \surr_t (M_{t-1-H:t}) \label{eq:regret_decomposition2}\\
    &\quad + \sum_{t=0}^{T-1} \surr_t (M_{t-1-H:t})  - \min_{M\in \calM} \sum_{t=0}^{T-1} f_t (M) , \label{eq:regret_decomposition3}
\end{align}
where $ K^* \in \calK (\kappa, \gamma) $ is an optimal state-feedback gain in hindsight.
\end{subequations}
In what follows, we introduce some lemmas for bounding \eqref{eq:regret_decomposition1}--\eqref{eq:regret_decomposition3}.

\subsection{Sufficiency of Disturbance-Action Policies}
It is known that given a state-feedback control policy $ u_t = -K^*x_t $, the disturbance-action policy~\eqref{eq:disturbance-action} with an appropriate parameter approximates it well for bounded noise~\citep[Lemma~5.2]{Agarwal2019}.
The following lemma shows that even when the noise distribution is unbounded, the disturbance-action policy is sufficient for approximating a given linear state-feedback control policy.

\begin{lemma}\label{lem:sufficiency}
    Suppose that Assumptions~\ref{ass:cost} and \ref{ass:noise_4}-(i) hold.
    Fix any two $ (\kappa,\gamma) $-strongly stable matrices $ K $ and $ K^* $ and define $ M_* = \{M_*^{[0]},\ldots,M_*^{[H-1]}\} $ as
    \begin{equation}\label{eq:M*}
        M_*^{[i]} := (K - K^*) (A-BK^*)^i .
    \end{equation}
    Then, $ M_* \in \calM $.
    In addition, when $ H = \lceil 2\gamma^{-1} \log T   \rceil $, it holds that for any $ C > 0 $ and $ \ft \ge 3 $, with probability at least $ 1- \frac{38\sigma_{w}\kappa_B^2 \kappa^6  }{\coe \gamma^2} $,
    \begin{align}
         &\sum_{t=0}^{\ft-1} \left| c_t (x_t^K(M_*), u_t^K (M_*)) - c_t (x_t^{K^*} , u_t^{K^*} ) \right| \le 2 G_c \coe^2 (\log\ft)^2 . \label{eq:M*_bound}
    \end{align}
\end{lemma}
\begin{proof}[Proof sketch]
    The detailed proof is given in Appendix~\ref{app:sufficiency}. By the triangle inequality, we have
\begin{align}
    & \sum_{t=0}^{\ft-1} \left| c_t (x_t^K(M_*), u_t^K (M_*)) - c_t (x_t^{K^*} , u_t^{K^*} ) \right|  \nonumber\\
    &\le \underbrace{ \sum_{t=0}^{\ft-1} \left| c_t (x_t^K(M_*), u_t^K (M_*)) - c_t (x_t^{K^*} , u_t^K (M_*)) \right|  }_{=: \term_{1}} + \underbrace{ \sum_{t=0}^{\ft-1}\left| c_t (x_t^{K^*} , u_t^K (M_*)) - c_t (x_t^{K^*} , u_t^{K^*} ) \right| }_{=: \term_{2}} . \label{eq:L1_L2_sketch}
\end{align}
By the convexity of $ c_t $ and Assumption~\ref{ass:cost}, the first term of \eqref{eq:L1_L2_sketch} can be bounded as follows:
\begin{align}
    \term_{1} &\le \sum_{t=0}^{\ft-1} \max \left\{ \| \nabla_x c_t(x_t^K(M_*), u_t^K(M_*)) \|, \| \nabla_x c_t(x_t^{K^*}, u_t^K(M_*)) \|  \right\} \| x_t^K (M_*) - x_t^{K^*} \| \nonumber\\
    &\le G_c \left( \sum_{t=0}^{\ft-1} \left(\| x_t^K (M_*) \| + \|x_t^{K^*} \| \right)  \right) \sum_{t=0}^{\ft-1}\| x_t^K (M_*) - x_t^{K^*} \| . \label{eq:term1_bound_sketch}
\end{align}
We bound each component of the right-hand side of \eqref{eq:term1_bound_sketch}.
In this proof sketch, we only focus on the last component.
By Markov's inequality, for any $ \thr > 0 $,
\begin{align}
    \bbP \left( \sum_{t=0}^{\ft-1} \| x_t^K(M_*) - x_t^{K^*} \| \ge \frac{\thr}{\ft}  \right) \le \frac{\ft}{\thr} \bbE\left[ \sum_{t=0}^{\ft-1} \| x_t^K(M_*) - x_t^{K^*} \| \right] . \label{eq:markov_2_sketch}
\end{align}
For bounding the right-hand side, the following expressions of $ x_t $ are useful:
\begin{align}
    x_{t}^K (M_*) &= 
         \sum_{i=0}^H \tilA_{K^*}^i w_{t-1-i} + \sum_{i=H+1}^{t-1} \Psi_{t-1,i}^{K,t-1}(M_*) w_{t-1-i} ,  && \forall t \ge 0 , \nonumber\\
    x_{t}^{K^*} &=  \sum_{i=0}^{t-1} \tilA_{K^*}^i w_{t-1-i}  ,  &&  \forall t \ge 0 . \nonumber
\end{align}
Then, we can bound $ \bbE[ \| x_{t}^K (M_*) - x_{t}^{K^*} \| ] $ as
\begin{align}
    \bbE\left[ \| x_{t}^K (M_*) - x_{t}^{K^*} \| \right] &=
     \bbE\left[ \left\|\sum_{i=H+1}^{t-1} \left(  \Psi_{t-1,i}^{K,t-1} (M_*) - \tilA_{K^*}^i \right) w_{t-1-i}  \right\|     \right] \nonumber\\
     &\le \begin{cases} \frac{2\sigma_{w}(H+1)\kappa_B^2 \kappa^5 (1-\gamma)^H }{\gamma} , & \forall t \ge H +2,  \\
        0, & \forall t \in \bbra{0,H+1} ,
     \end{cases} \nonumber
\end{align}
where we used H\"{o}lder's inequality and the fact $ \|\Psi_{t,i}^{K,h} (M_{*}) \| \le (2H+1)\kappa_B^2 \kappa^5 (1-\gamma)^{i-1} $.
Therefore, for any $ \ft \ge 3 $ and $ \coe > 0 $, by letting $ H = \lceil 2\gamma^{-1} \log \ft  \rceil $ and $ \thr = \coe \log \ft $ in \eqref{eq:markov_2_sketch}, we obtain
\begin{align}
    \bbP \left( \sum_{t=0}^{\ft-1} \| x_t^K(M_*) - x_t^{K^*} \| \ge \frac{\coe \log \ft}{\ft}  \right) 
    \le \frac{2\sigma_{w}(2\gamma^{-1} + \frac{2}{\log \ft})\kappa_B^2 \kappa^5 }{\coe \gamma}
    \le \frac{8\sigma_{w}\kappa_B^2 \kappa^5 }{\coe \gamma^2}. \nonumber
\end{align}
The other components of \eqref{eq:term1_bound_sketch} can be bounded similarly, and consequently, we obtain
\begin{align}
    \bbP\left( \term_1 \le G_c \coe^2 (\log \ft)^2  \right) \ge 1 - \frac{17\sigma_{w} \kappa_B^2 \kappa^5}{\coe\gamma^2} , ~~ \forall \coe > 0, \ \ft \ge 3 . \nonumber
\end{align}
A similar argument applies to bounding $ \term_2 $ in \eqref{eq:L1_L2_sketch}, and we obtain the desired result.
\end{proof}

\subsection{Difference between Surrogate and Actual Costs}
We next show how well the surrogate cost approximates the actual cost.
The proof is provided in Appendix~\ref{app:approximation} using similar techniques as in the proof of Lemma~\ref{lem:sufficiency}.

\begin{lemma}\label{thm:approximation}
    Suppose that Assumptions~\ref{ass:cost} and \ref{ass:noise_4}-(i) hold and $ M_t \in \calM $ for any $ t\in \bbra{0,T-1} $.
    Let $ K $ be $ (\kappa,\gamma) $-strongly stable and $ H = \lceil 2\gamma^{-1} \log T   \rceil $.
    Then, for any $ \coe > 0 $ and $ \ft \ge 3 $, with probability at least $ 1 - \frac{46 \sigma_w  \kappa_B^2 \kappa^{8} }{\coe \gamma^2(1-\gamma)} $, it holds that
    \begin{align}
        &\sum_{t=0}^{\ft-1} \left| \surr_t (M_{t-1-H:t}) -c_t (x_t^K (M_{0:t-1}), u_t^K(M_{0:t}))      \right| \le 2 G_c \coe^2 (\log\ft)^2 .  \label{eq:surrogate_actual_bound}
    \end{align}
\end{lemma}

\subsection{Bound by Online Convex Optimization with Memory}\label{subsec:OCO}
This subsection is devoted to bounding \eqref{eq:regret_decomposition3}.
Although the regret bound \eqref{eq:bound_memory_convention} for OCO with memory can be directly utilized for the regret analysis under bounded noise~\citep{Agarwal2019}, it is unlikely to be applicable to bounding \eqref{eq:regret_decomposition3} by $ \wtilde{O}(\sqrt{\ft}) $ for unbounded noise. This is because as observed from the proof of Lemma~\ref{lem:gradient}, $  \sup_{t\in \bbra{0,T-1},M\in \calM} \| \nabla f_t (M) \|_{\rm F} $, which corresponds to $ G_\ell $ in \eqref{eq:bound_memory_convention}, cannot be bounded by $ O({\rm poly} (\log \ft)) $ with high probability without additional assumptions on the noise. Therefore, instead of \eqref{eq:bound_memory_convention}, we employ a more fundamental bound for OCO with memory, and accordingly, the learning rate $ \eta_t = \Theta (1/\sqrt{\ft}) $ used for the bounded noise~\citep{Agarwal2019} is modified to $ \eta_t = \Theta ((\sqrt{\ft} (\log \ft)^3)^{-1}) $ for unbounded noise.
Specifically, by the standard argument for OCO with memory~\citep{Anava2015}, the regret for the surrogate cost $ \surr_t $ with OGD can be bounded as follows. The proof is shown in Appendix~\ref{app:OCO}.
\begin{proposition}\label{prop:OCO_regret}
    Let $ \{M_t \}_{t=-H-1}^{0} $ be an arbitrary sequence such that $ M_t \in \calM $ for any $ t \in \bbra{-H-1,0} $ and $ \{M_t\}_{t=1}^{\ft-1} $ be given by OGD with learning rate $\eta_t = \eta > 0$:
    \begin{align}\label{eq:ogd}
        M_{t+1} &= \Pi_{\calM} (M_t - \eta \nabla f_t (M_t)) := \argmin_{M' \in \calM} \| M_t - \eta \nabla f_t (M_t) - M' \|_{\rm F} , ~~ t \in \bbZ_{\ge 0} .
    \end{align}
    Assume that there exists $ L_{\rmc} > 0 $ such that for any $ k \in \bbra{0,H+1} $, $ t\in \bbra{0,T-1} $, and $ M_{t-1-H},\ldots,M_t, \widecheck{M}_{t-k}\in \calM  $,
    \begin{align}
        &\left| \surr_t (M_{t-1-H:t}) -  \surr_t (M_{t-1-H:t-k-1},\wcheck{M}_{t-k},M_{t-k+1:t}) \right| \le  L_\rmc d_M (M_{t-k},\wcheck{M}_{t-k}), \label{eq:lipschitz}\\
        & d_M (M_{t-k},\wcheck{M}_{t-k}) := \left\|
            \begin{bmatrix}
              M_{t-k}^{[0]} - \wcheck{M}_{t-k}^{[0]}  \\
              \vdots \\
              M_{t-k}^{[H-1]} - \wcheck{M}_{t-k}^{[H-1]} 
            \end{bmatrix} \right\|_{\rm F} .  \label{eq:distance_M}
    \end{align}
    Then, for $ D := \frac{4\kappa_B \kappa^3 \sqrt{\dime}}{\gamma} $ and for any $ M\in \calM $, it holds that
    \begin{align}
        \sum_{t=0}^{T-1} \surr_t (M_{t-1-H:t}) 
        - 
        \sum_{t=0}^{T-1} f_t (M)
       &\leq L_\rmc \eta \sum_{t=0}^{\ft-1} \sum_{i=1}^{\min\{H+1,t\}} \sum_{k=1}^i
        \|
            \nabla_M f_{t-k}(M_{t-k})
        \|_{\rm F} + \frac{1}{2\eta} D^2 \nonumber\\
        &\quad + \frac{\eta}{2} \sum_{t=0}^{\ft-1} \| \nabla_M f_t (M_t) \|_{\rm F}^2 
        . \label{eq:general_ogd_bound}
    \end{align}
    \hfill $ \diamondsuit $
\end{proposition}

For bounded noise, the gradient $ \nabla_M f_{t} (M_t) $ and the coordinate-wise Lipschitz constant $ L_\rmc $ are bounded uniformly in $ t $ \citep{Agarwal2019}, and consequently, OGD with $ \eta = \Theta(1/\sqrt{\ft}) $ attains $ O(\sqrt{\ft}) $ regret by \eqref{eq:bound_memory_convention}. On the other hand, the gradient and the Lipschitz constant are not bounded under unbounded noise. 
Instead, for the coordinate-wise Lipschitz constant, we have the following result.
The proof is given in Appendix~\ref{app:lipshitz}.
\begin{lemma}\label{lem:lipshitz}
    Suppose that Assumptions~\ref{ass:cost} and \ref{ass:noise_4}-(i) hold, and let $ K $ be $ (\kappa,\gamma) $-strongly stable and $ H = \lceil 2\gamma^{-1} \log T   \rceil $.
    Then, for any $ \ft \ge 3 $, $ t \in \bbra{0,\ft-1} $, $ k\in \bbra{0,H+1} $, $ \coe > 0 $, and $ M_{t-1-H},\ldots,M_t, \widecheck{M}_{t-k}\in \calM  $, with probability at least $ 1 - \frac{30\sigma_w \kappa_B^2 \kappa^{6}}{\coe \gamma^2 (1-\gamma)}  $, it holds that
    \begin{align}
        &\left| \surr_t (M_{t-1-H:t}) -  \surr_t (M_{t-1-H:t-k-1},\wcheck{M}_{t-k},M_{t-k+1:t}) \right| \nonumber\\
        &\qquad\qquad  \le  2 G_c \coe^2 \log \ft \left(2 \gamma^{-1}\log \ft + 1\right)^{1/2} d_M(M_{t-k},\wcheck{M}_{t-k}) , \nonumber
    \end{align}
    where $ d_M $ is defined in \eqref{eq:distance_M}.
    \hfill $ \diamondsuit $
\end{lemma}

We emphasize that the bounds derived so far (Lemmas~\ref{lem:sufficiency},~\ref{thm:approximation}, and \ref{lem:lipshitz}) are logarithmic in $ \ft $. One of the reasons this can be attained is that they evaluate the differences such as between the surrogate and actual costs. For bounding the gradient of the surrogate cost, this is not the case, and the unboundedness of the noise has a significant impact. Nevertheless, the following result together with Theorem~\ref{thm:regret} reveals that the finiteness of the fourth moment of the noise sufficiently bounds the first and third terms of \eqref{eq:general_ogd_bound} for attaining the $ \wtilde{O}(\sqrt{\ft}) $ regret.
\begin{lemma}\label{lem:gradient}
    Suppose that Assumptions~\ref{ass:cost} and \ref{ass:noise_4}-(ii) hold, and let $ K $ be $ (\kappa,\gamma) $-strongly stable and $ H = \lceil 2\gamma^{-1} \log T   \rceil $.
    Then, for any $ \coe > 0 $, $ \ft \ge 3 $, and $ M\in \calM $, it holds that
\begin{align}
    &\bbP\left( \sum_{t=0}^{\ft-1} \|\nabla_M f_t (M) \|_{\rm F}^2 \le \coe  \ft (\log \ft)^4  \right) \ge 1- \frac{26244 \sigma_w^4 \dime^2 G_c^2  \kappa_B^6 \kappa^{18} }{\coe \gamma^8 (1-\gamma)^4}, \label{eq:sum_grad}\\
    &\bbP\left( \sum_{t=0}^{\ft-1} \sum_{i=1}^{\min\{H+1,t\}} \sum_{k=1}^i
    \|
        \nabla_M f_{t-k}(M)
    \|_{\rm F} \le \coe \ft (\log \ft)^{\frac{9}{2}} \right) \ge 1- \frac{39366 \sigma_w^2\dime^2 G_c \kappa_B^3 \kappa^{9} }{\coe  \gamma^{13/2} (1-\gamma)^2} . \nonumber
\end{align}
\end{lemma}
\begin{proof}[Proof sketch]
    We only sketch the proof of \eqref{eq:sum_grad}, and the detailed proof is deferred to Appendix~\ref{app:gradient}.
    For $ r\in \bbra{0,H-1} $, $p\in \bbra{1,n_\rmu}$, and $q\in \bbra{1,n_\rmx} $, let us bound $ \nabla_{M_{p,q}^{[r]}} f_t (M) $, where $ M_{p,q}^{[r]} $ denotes the $ (p,q) $ entry of $ M^{[r]} $.
For notational simplicity, we drop the superscript of $ y_t^K $ and $v_t^K $.
By using the chain rule and Assumption~\ref{ass:cost}, we obtain
\begin{align}
  \left|\nabla_{M_{p,q}^{[r]}} f_t (M) \right| &= \left| \nabla_x c_t (y_t(M),v_t(M))^\top \frac{\partial y_t(M)}{\partial M_{p,q}^{[r]}} + \nabla_u c_t (y_t (M), v_t(M))^\top \frac{\partial v_t (M)}{\partial M_{p,q}^{[r]}} \right| \nonumber\\
  &\le G_c  \|y_t (M)\|   \left\|\frac{\partial y_t(M)}{\partial M_{p,q}^{[r]}}\right\|  + G_c  \| v_t(M) \|  \left\| \frac{\partial v_t (M)}{\partial M_{p,q}^{[r]}} \right\|  . \label{eq:grad_pq_bound_sketch}
\end{align}
By \eqref{eq:surrogate_y} and the bound of $ \Psi $ used in the proof of Lemma~\ref{lem:sufficiency}, for any $ M \in \calM $, we have
\begin{align}
    \|y_t(M) \| &=  \left\| \sum_{i=0}^{2H} \Psi_{t-1,i}^{K,H} (M) w_{t-1-i} \right\| \le \sum_{i=0}^{2H} (2H+1) \kappa_B^2 \kappa^5 (1-\gamma)^{i-1} \|w_{t-1-i}\| , \nonumber\\
    \left\| \frac{\partial y_t(M)}{\partial M_{p,q}^{[r]}} \right\| &= \left\| \sum_{i=0}^{2H} \sum_{j=0}^H \frac{\partial \tilA_K^j B M^{[i-j-1]}}{\partial M_{p,q}^{[r]}} w_{t-1-i} \bone_{i-j\in \bbra{1,H}}          \right\| \nonumber\\
    &\le \sum_{i=r+1}^{r+H+1} \kappa_B \kappa^2 (1-\gamma)^{i-r-1} \|w_{t-1-i} \| . \label{eq:dydM_sketch}
\end{align}
Similarly, we can bound $ v_t $ and its derivative.
Consequently, for bounding $ \sum_{t=0}^{\ft-1} \| \nabla_M f_t (M_t) \|_{\rm F}^2 $, we have
\begin{align}
   &\bbE\left[\left|\nabla_{M_{p,q}^{[r]}} f_t (M) \right|^2\right] \le 4 G_c^2 (2H+3)^2 \kappa_B^6 \kappa^{18} \nonumber\\
    &\quad \times \bbE\left[ \left(  \sum_{i=0}^{2H} (1-\gamma)^{i-1} \|w_{t-1-i}\| \right)^4   \right]^{1/2} \bbE\left[ \left( \sum_{i=r}^{r+H+1}  (1-\gamma)^{i-r-1} \|w_{t-1-i}\| \right)^4  \right]^{1/2} , \label{eq:nabla_square_bound_sketch}
\end{align}
where we used H\"{o}lder's inequality.
The finiteness of the fourth moment of $ w_t $ is used to bound \eqref{eq:nabla_square_bound_sketch} as follows:
\begin{align}
    \bbE\left[ \left(  \sum_{i=0}^{2H} (1-\gamma)^{i-1} \|w_{t-1-i}\| \right)^4   \right] &\le \frac{1}{\gamma^4 (1-\gamma)^4} \bbE\left[ \max_{i\in \bbra{0,2H}} \|w_{t-1-i}\|^4 \right] \nonumber\\
    &\le \frac{2H+1}{\gamma^4 (1-\gamma)^4} \max_{i\in \bbra{0,2H}} \bbE\left[ \|w_{t-1-i}\|^4\right] \le \frac{(2H+1)\sigma_w^4}{\gamma^4 (1-\gamma)^4} . \nonumber
\end{align}
The second component of \eqref{eq:nabla_square_bound_sketch} can be bounded similarly, and we obtain \eqref{eq:sum_grad}.
\end{proof}

By Proposition~\ref{prop:OCO_regret} and Lemmas~\ref{lem:lipshitz} and \ref{lem:gradient}, we obtain the $ \wtilde{O}(\sqrt{\ft}) $ high-probability bound for OGD.
\begin{proposition}\label{prop:OCO_regret_general}
    Suppose that Assumptions~\ref{ass:cost} and \ref{ass:noise_4}-(ii) hold, and let $ K $ be $ (\kappa,\gamma) $-strongly stable and $ \{M_t \}_{t=-H-1}^{0} $ be an arbitrary sequence such that $ M_t \in \calM $ for any $ t \in \bbra{-H-1,0} $.
    Define $ D := \frac{4\kappa_B \kappa^3 \sqrt{\dime}}{\gamma} $ and $ \sigma_w^{[1,4]} := \max\{\sigma_w, \sigma_w^4\} $.
    Then, if a sequence $ \{M_t\}_{t=1}^{\ft-1} $ is given by OGD \eqref{eq:ogd} with learning rate $ \eta_t \equiv (\sqrt{\ft} (\log \ft)^3 )^{-1} $, it holds that for any $ \coe > 0 $, $ \ft \ge 3 $, and $ M\in \calM $, with probability at least $ 1 -  \frac{65640 \sigma_w^{[1,4]} \dime^2 G_c^2  \kappa_B^6 \kappa^{18} }{\coe \gamma^8 (1-\gamma)^4} $,
    \begin{align}
        &\sum_{t=0}^{T-1} \surr_t (M_{t-H:t}) 
        - 
        \sum_{t=0}^{T-1} f_t (M)
       \leq  \left( \frac{2\sqrt{3}G_c \coe^3}{\sqrt{\gamma}} +  \frac{D^2}{2} \right)  \sqrt{\ft}(\log \ft)^{3}   + \frac{\coe}{2} \sqrt{\ft} \log \ft . \nonumber
    \end{align}
\end{proposition}

\subsection{Proof of Theorem~\ref{thm:regret}}\label{subsec:proof_regret}
Now, we are ready to prove Theorem~\ref{thm:regret}. 
    The second and third lines \eqref{eq:regret_decomposition2},~\eqref{eq:regret_decomposition3} of the regret can be bounded by Lemma~\ref{thm:approximation} and Proposition~\ref{prop:OCO_regret_general}, respectively.
    For the first line \eqref{eq:regret_decomposition1}, we have
    \begin{align}
        \text{\eqref{eq:regret_decomposition1}} &\le\sum_{t=0}^{T-1} f_t (M_*) - \sum_{t=0}^{T-1} c_t (x_t^K(M_*),u_t^K(M_*))  \nonumber\\
        &\quad +  \sum_{t=0}^{T-1} c_t (x_t^K(M_*),u_t^K(M_*))    -  \sum_{t=0}^{T-1} c_t (x_t^{K^*},u_t^{K^*})  ,\nonumber
    \end{align}
    which can be bounded by Lemmas~\ref{lem:sufficiency} and \ref{thm:approximation}. Consequently, for any $ \coe > 0 $, we have
    \begin{align}
        &\bbP\Biggl( \regret_\ft \le  \left( \frac{2\sqrt{3}G_c \coe^3}{\sqrt{\gamma}} + \frac{D^2}{2} \right) \sqrt{\ft}(\log \ft)^{3}  + \frac{\coe}{2} \sqrt{\ft} \log \ft + 6 G_c \coe^2 (\log\ft)^2 \Biggr) \nonumber\\
        & \ge 1 - \frac{38\sigma_{w}\kappa_B^2 \kappa^6  }{\coe \gamma^2}  - \frac{46 \sigma_w  \kappa_B^2 \kappa^{8} }{\coe \gamma^2(1-\gamma)} - \frac{65640 \sigma_w^{[1,4]} \dime^2 G_c^2  \kappa_B^6 \kappa^{18}}{\coe \gamma^8 (1-\gamma)^4} .\nonumber
    \end{align}
    This completes the proof.

\section{Regret Analysis for Strongly Convex Costs}

In this section, we focus on strongly convex and smooth cost $ c_t $ and show that an $ O({\rm poly}(\log T)) $ regret can be achieved. Specifically, we make the following assumption, as in \citet{Simchowitz2020improper}.

\begin{assumption}\label{ass:strong_convex_cost}
    The cost function $ c_t : \bbR^{n_\rmx} \times \bbR^{n_\rmu} \rightarrow \bbR $ is twice differentiable for any $ t \in \bbra{0,\ft-1} $.
    In addition, there exist $ \alpha > 0 $ and $\beta > 0 $ such that the cost function $ c_t $ satisfies
    \begin{align}
        \alpha I \preceq \nabla_{x,u}^2 c_t (x,u) \preceq \beta I 
    \end{align}
    for any $ x\in \bbR^{n_\rmx} $, $ u\in \bbR^{n_\rmu} $, and $ t \in \bbra{0,T-1} $.
    \hfill $ \diamondsuit $ 
\end{assumption}
In Lemma~\ref{lem:strong_convex}, we will see that $ \alpha $-strong convexity of $ c_t $ leads to the strong convexity of the {\em expected} surrogate cost $ \bar{f}_t (M) := \bbE[f_t(M)] $, which enables OGD to attain logarithmic {\em pseudo}-regret $ \bbE[J_T (\{u_t\})] - \min_{\{u_t^*\} \in \calU} \bbE[J_T(\{u_t^*\})] = O({\rm poly} \log T)$ as in \citet{Agarwal2019log}. However, to establish a high-probability logarithmic regret bound, we need to require an additional assumption that $ c_t $ is $ \beta $-smooth, since $ f_t $ is not strongly convex even if $ c_t $ is.

In Section~\ref{sec:regret}, we have shown that \eqref{eq:regret_decomposition1} and \eqref{eq:regret_decomposition2} of the decomposition of the regret can be bounded by $ O({\rm poly} (\log \ft)) $ with high probability for general convex costs. Hence, we focus on \eqref{eq:regret_decomposition3} in this section.
Unlike the $ \wtilde{O}(\sqrt{\ft}) $ regret bound in the previous section, the finiteness of the $ p $th moment of the noise is not sufficient for attaining logarithmic regret. 
In this work, we make the following assumptions on the noise.
\begin{assumption}\label{ass:subgauss}
    The noise process $ \{w_t\} $ is an independent sequence and satisfies $ \bbE[w_t] = 0 $ for any $ t \in \bbZ_{\ge 0} $.
    In addition, for $ w_t = [w_{t,1}, \ldots , w_{t,\dime_\rmx}]^\top $, there exist $ \sigma_w \ge 0 $ and $ \usigma > 0 $ such that
    \begin{align}
        \log \bbE\left[ \Ee^{\lambda w_{t,i}} \right] &\le \frac{\lambda^2 \sigma_w^2}{2\dime_\rmx}, ~ &&\forall \lambda \in \bbR , \ t \ge 0, \ i \in \bbra{1,\dime_\rmx} , \label{eq:subgauss_def} \\
        \Sigma_t := \bbE[w_t w_t^\top] &\succeq \usigma^2 I , \ &&\forall t \ge 0 . \label{eq:nondegenerate}
    \end{align}
    \hfill $ \diamondsuit $
\end{assumption}
The condition~\eqref{eq:subgauss_def} means the sub-Gaussianity of the noise~\citep{Van2014}.
If $ w_t $ satisfies \eqref{eq:subgauss_def}, $ w_t $ has a finite $ p $th moment for any $ p \ge 1 $. In particular, it holds that $ \bbE[\|w_t\|^2] \le \sigma_w^2 $ \citep[Problem~3.1]{Van2014}, which means that Assumption~\ref{ass:noise_4}-(i) in Lemmas~\ref{lem:sufficiency}, \ref{thm:approximation}, and \ref{lem:lipshitz} is satisfied for the same scale parameter $ \sigma_w $. The non-degeneracy of the noise \eqref{eq:nondegenerate} is introduced to ensure that the expectation of the surrogate cost is strongly convex~\citep{Agarwal2019log}; see Lemma~\ref{lem:strong_convex}.

To bound \eqref{eq:regret_decomposition3}, we use the following result.
\begin{proposition}[{\citealt[Theorem~8]{Simchowitz2020improper}}]\label{prop:log_regret}
    Assume that $ \surr_t $ satisfies \eqref{eq:lipschitz} with the coordinate-wise Lipschitz constant $ L_\rmc $. Assume also that there exist constants $ L_\rmf > 0 $ and $ \beta' > 0 $ such that $ f_t (M) := \surr_t (M,\ldots,M) $ satisfies for any $ t\in \bbra{0,\ft-1} $ and $ M,\wcheck{M} \in \calM $,
    \begin{align}
        | f_t (M) -  f_t (\wcheck{M}) | &\le  L_{\rm f}  \| M - \wcheck{M}  \|_{\rm F} , \label{eq:lipschitz2} \\
        \|\nabla_M^2 f_t(M) \|_{\rm F} &\le \beta'  .
    \end{align}
    Assume further that $ \bar{f}_t (M) := \bbE[f_t(M)] $ is $ \alpha' $-strongly convex on $ \calM $.
    Then, if $ M_0 \in \calM $ and a sequence $ \{M_t\}_{t=1}^{\ft-1} $ is given by OGD \eqref{eq:ogd} with learning rate $ \eta_t = 3/(\alpha' (t+1)) $, it holds that for any $ \delta \in (0,1] $ and $ M \in \calM $, with probability at least $ 1-\delta $,
    \begin{align}
        \sum_{t = H+3}^{\ft-1} \surr_t (M_{t-1-H:t}) -  \sum_{t = H+3}^{\ft-1} f_t (M) &\lesssim \alpha' D^2 H_+ + \frac{L_{\rm f}^2 H_+}{\alpha'} \log \left( \frac{1 + \log (\Ee + \alpha' D^2)}{\delta} \right) \nonumber\\
        &\quad + \frac{(2L_{\rm c} + \dime^2 L_\rmf ) H_+ + \beta'}{\alpha'} L_{\rmf} H_+ \log \ft , \label{eq:regret_oco_log}
    \end{align}
    where $ H_+ := H + 2 $.
    \hfill $ \diamondsuit $
\end{proposition}

In the rest of this section, we bound the constants $ L_{\rmf} $ and $ \beta' $ and obtain $ \alpha' $ in Proposition~\ref{prop:log_regret}.
First, it is known that under the strong convexity of $ c_t $ and the non-degeneracy of the noise, $ \bar{f}_t $ is strongly convex as follows.\footnote{Although Lemma~4.2 of \citet{Agarwal2019log} assumes that $ w_t $ is bounded, its proof still works without the boundedness.}
\begin{lemma}[{\citealt[Lemma~4.2]{Agarwal2019log}}]\label{lem:strong_convex}
    Suppose that Assumptions~\ref{ass:cost},~\ref{ass:strong_convex_cost}, and \ref{ass:subgauss} hold and $ K $ is $ (\kappa,\gamma) $-diagonally strongly stable. Then, $ \bar{f}_t(M) := \bbE[f_t(M)]$ is $ \wtilde{\alpha} $-strongly convex on $ \calM $ where
    \begin{align}
        \wtilde{\alpha} := \frac{\alpha \usigma^2 \gamma^2}{36 \kappa^{10}}  . \tag*{$\diamondsuit$}
    \end{align}
\end{lemma}

Similar to the coordinate-wise Lipschitz constant (Lemma~\ref{lem:lipshitz}), we have the following result for $ L_\rmf $. The proof is given in Appendix~\ref{app:lipshitz}.
\begin{lemma}\label{lem:lipschitz_global}
    Suppose that Assumptions~\ref{ass:cost} and \ref{ass:noise_4}-(i) hold, and let $ K $ be $ (\kappa,\gamma) $-strongly stable and $ H = \lceil 2\gamma^{-1} \log T   \rceil $. 
    Then, for any $ \ft \ge 3 $, $ t \in \bbra{0,\ft-1} $, $ M, \wcheck{M} \in \calM $, and $ \coe > 0 $, with probability at least $ 1 - \frac{41\sigma_w \kappa_B^2 \kappa^{6}}{\coe \gamma^2 (1-\gamma)} $, it holds that
    \begin{align}
        | f_t (M) -  f_t (\wcheck{M}) | &\le 2 G_c \coe^2 (\log \ft)^2 \left(2\gamma^{-1}\log \ft + 1\right)^{1/2}   \|M - \wcheck{M}\|_{\rm F} . \tag*{$\diamondsuit$}
    \end{align}
\end{lemma}

The sub-Gaussianity of the noise and the $ \beta $-smoothness of the cost functions are utilized to bound the Hessian of $ f_t $ by $ O({\rm poly} (\log \ft)) $.
\begin{proposition}\label{prop:hessian}
    Suppose that Assumptions~\ref{ass:strong_convex_cost} and \ref{ass:subgauss} hold, and let $ K $ be $ (\kappa,\gamma) $-strongly stable and $ H = \lceil 2\gamma^{-1} \log T   \rceil $. 
    Then, for any $ \ft \ge 3 $, $ t \in \bbra{0,\ft-1} $, $ \coe > 0 $, and $ M \in \calM $, with probability at least $ 1 -\frac{2\sigma_w\sqrt{2 \dime_\rmx (1+ \log \dime_\rmx)  }}{\coe}  $, it holds that
    \begin{align}
        \max_{t\in \bbra{0,\ft-1}} \left\| \nabla_{M}^2 f_t (M) \right\|_{\rm F} \le  \frac{6\kappa_B \kappa^3 \beta \dime^2   \coe (\log \ft)^{3/2}}{\gamma^2(1-\gamma)} . \tag*{$\diamondsuit$}
    \end{align}
\end{proposition}
\begin{proof}[Proof sketch]
    The detailed proof is shown in Appendix~\ref{app:smoothness}.
    For notational simplicity, we drop the superscript $ K $ of $ y_t^{K} $ and $ v_t^{K} $. 
    By the chain rule and the $ \beta $-smoothness of $ c_t $ (Assumption~\ref{ass:strong_convex_cost}), we obtain
\begin{align}
    \left\| \nabla_{M}^2 f_t (M) \right\|_{\rm F}
    &\le  \left\| \begin{bmatrix}
        \calJ_{y_t} \\ \calJ_{v_t}
    \end{bmatrix}\right\|_{\rm F}^2
     \|\nabla_{x,u}^2 c_t (y_t (M), v_t (M)) \|_{\rm F}
      \le \beta \left\| \begin{bmatrix}
        \calJ_{y_t} \\ \calJ_{v_t}
    \end{bmatrix}\right\|_{\rm F}^2 , \label{eq:hessian_sketch}
\end{align}
where $ \calJ_{y_t} $ and $ \calJ_{v_t} $ are the Jacobian matrices of $ y_t $ and $ v_t $ with respect to $ M $, respectively.
Here, we have
\begin{align}
    \left\|\begin{bmatrix}
        \calJ_{y_t} \\ \calJ_{v_t}
    \end{bmatrix}\right\|_{\rm F}^2  = \sum_{p=1}^{\dime_\rmu} \sum_{q=1}^{\dime_\rmx} \sum_{r=0}^{H-1} \left( \left\| \frac{\partial y_t(M)}{\partial M_{p,q}^{[r]}}  \right\|^2 + \left\| \frac{\partial v_t(M)}{\partial M_{p,q}^{[r]}}  \right\|^2 \right) . \label{eq:jacobi_bound_sketch}
\end{align}
Moreover, we have the following maximal inequality for sub-Gaussian noise $ w_t $:
\begin{align}
    \bbP\left(\max_{t\in \bbra{0,\ft-1}} \|w_t\| \ge \thr \right) 
    \le \frac{2  \sigma_w\sqrt{2 \dime_\rmx \log (\ft \dime_\rmx)} }{\thr} , ~~ \forall \thr > 0, \ \ft \ge 2 .
\end{align}
Hence, it follows from the bound \eqref{eq:dydM_sketch} of $ \partial y_t / \partial M_{p,q}^{[r]} $ and its counterpart for $ v_t $ that for any $ \thr > 0 $, $ \ft \ge 2 $, and $  M \in \calM $,
\begin{align}
    &\bbP\Biggl( \max_{t\in \bbra{0,\ft-1}, p\in \bbra{1,\dime_u}, q \in \bbra{1,\dime_\rmx}, r\in \bbra{0,H-1}} \left\| \frac{\partial y_t(M)}{\partial M_{p,q}^{[r]}}  \right\| \le \frac{\kappa_B \kappa^2 \thr}{\gamma}, ~~ \max_{t,p,q,r} \left\| \frac{\partial v_t(M)}{\partial M_{p,q}^{[r]}}  \right\| \le \frac{\kappa_B \kappa^3 \thr}{\gamma(1-\gamma)} \Biggr) \nonumber\\
    &\ge  1 - \frac{2\sigma_w\sqrt{2 \dime_\rmx\log (\ft \dime_\rmx)} }{\thr} . \nonumber
\end{align}
By combining this with \eqref{eq:hessian_sketch} and \eqref{eq:jacobi_bound_sketch} and by a straightforward calculation, we obtain the desired result.
\end{proof}

By Proposition~\ref{prop:log_regret} with Lemmas~\ref{lem:strong_convex} and \ref{lem:lipschitz_global} and Proposition~\ref{prop:hessian}, we finally establish the logarithmic regret bound. The proof is shown in Appendix~\ref{app:log_regret}
\begin{theorem}\label{thm:log_regret}
    Suppose that Assumptions~\ref{ass:cost},~\ref{ass:strong_convex_cost}, and \ref{ass:subgauss} hold, and let $ K $ be $ (\kappa,\gamma) $-diagonally strongly stable and $ H = \lceil 2\gamma^{-1} \log T   \rceil $. 
    Then, for Algorithm~\ref{alg:online_control} with $ \eta_t = 3/(\wtilde{\alpha} (t+1)) $, where $ \wtilde{\alpha} := \alpha \usigma^2 \gamma^2/(36 \kappa^{10}) $, it holds that for any $ T \ge 3 $ and $ \delta \in (0,1] $, with probability at least $ 1-\delta $,
    \begin{align}
         \regret_\ft^{\rm diag} &\lesssim \wtilde{\alpha}  D^2 H_\ft + \frac{ \bar{L}_{\ft,\delta}^2 H_\ft}{\wtilde{\alpha}} \log \left( \frac{1 + \log (\Ee + \wtilde{\alpha} D^2)}{\delta} \right) \nonumber\\
        &\quad+ \frac{ (2+ \dime^2) \bar{L}_{\ft,\delta} H_\ft  + \bar{\beta}_{\ft,\delta} }{\wtilde{\alpha}} \bar{L}_{\ft,\delta} H_\ft \log \ft \nonumber\\
        &\quad+ 2  D \bar{L}_{\ft,\delta} (H_\ft+1)^3 + 6 G_c C_\delta^2 (\log \ft)^{2},
    \end{align}
    where $ C_\delta := \frac{2}{\delta} \left( \frac{201 \sigma_w \kappa_B^2 \kappa^8}{\gamma^2 (1-\gamma)} + 2 \sigma_w \sqrt{2\dime_\rmx (1+ \log \dime_\rmx)} \right) $, $ H_\ft := \lceil 2\gamma^{-1} \log T   \rceil + 2 $, $ \bar{L}_{\ft,\delta} := \frac{4G_c C_\delta^2 (\log \ft)^{5/2}}{\sqrt{\gamma}} $, $ \bar{\beta}_{\ft,\delta} := \frac{6\kappa_B \kappa^3 \beta \dime^2  \coe_\delta  (\log \ft)^{3/2}}{\gamma^2 (1-\gamma)} $, and $ D := \frac{4\kappa_B \kappa^3 \sqrt{\dime}}{\gamma} $.
    \hfill $ \diamondsuit $
\end{theorem}
That is, the regret bound of $ O((\log \ft)^8) $ can be attained with high probability for strongly convex and smooth cost and sub-Gaussian noise.

\begin{remark}
    A high-probability regret bound typically guarantees that, with probability at least $ 1-\delta $, the regret bound holds with an additional overhead of at most $ O(\log (1/\delta)) $ compared to the pseudo-regret bound. In contrast, for the online control of linear systems under unbounded noise, we observe in Appendix~\ref{app:log_delta} that obtaining a high-probability bound with only an $ O(\log (1/\delta)) $ overhead is likely impossible.
    \hfill $ \diamondsuit $
\end{remark}

\section{Conclusion}\label{sec13}

In this paper, we tackled the online control problem under unbounded noise.
We only require noise to have a finite fourth moment to obtain the $ \wtilde{O} (\sqrt{\ft}) $ regret bound that holds with high probability.
Moreover, the strong convexity of costs and the sub-Gaussianity of the noise enable us to achieve a logarithmic regret bound.
In other words, we revealed that in online control, not only the curvature of the cost functions but also the statistical properties of the noise play a crucial role in improving the regret bound, highlighting a clear contrast to the standard OCO.
Consequently, our results broaden the applicability of online control.
An important direction of future work is to extend our results to online control problems without the knowledge of systems~\citep{Hazan2020nonstochastic} and to a partially observed setting~\citep{Simchowitz2020improper}, and we expect that our approach would be applicable to such problems.

\newpage

\begin{appendices}


\section{Possible Impossibility of an $ O(\log (1/\delta)) $ Overhead for a High-probability Regret Bound}\label{app:log_delta}
In this appendix, we observe that it would be impossible to obtain a regret bound that holds with probability $ 1-\delta $ with only an additional $ O(\log (1/\delta)) $ factor for online control of linear systems under unbounded noise. To see this, let us consider the simplest case where the system is given by $ x_{t+1} = u_t + w_t $ and the cost $ c_t(x_t,u_t) = c_t(x_t) $ depends only on the state $ x_t $. The noise $ w_t $ is assumed to satisfy $ \bbE[w_t] = 0 $ and $ \bbE[\|w_t \|^2] \le \sigma_w^2 < \infty $.
	Since the system is static in this case, the regret minimization reduces to the standard OCO without memory.
	Then, by the standard argument for OCO, the regret of OGD can be bounded as follows:
	\begin{align*}
		\regret_T := \sum_{t=0}^{T-1} c_t (x_{t}) - \min_{u_t \equiv  u\in \bbU} \sum_{t=0}^{T-1} c_t (x_{t}) 
		\le \frac{1}{2\eta} D^2 + \frac{\eta G_c^2}{2} \sum_{t=0}^{T-1}   \| u_t + w_t \|^2 ,
	\end{align*}
	where $ \bbU $ is the set of admissible control inputs, $ D := \sup_{u,v\in \bbU} \|u - v \| $ is the diameter of $ \bbU $, $ \eta $ is the learning rate of the OGD, $ \{u_t\} $ is given by OGD $ u_{t} = u_{t-1} - \eta \nabla c_{t-1} (u_{t-2} + w_{t-2}) $, and we used the assumption that $ \| \nabla c_t (x) \| \le G_c \|x\| $ for any $ x \in \bbR^{n_{\rm x}} $ (Assumption~\ref{ass:cost}).
	When we only assume the finiteness of the second moment of $ w_t $, the best we can do to bound the second term of the right-hand side would be to use Markov's inequality. For any $ C > 0 $, we have
	$
		\bbP\left( \sum_{t=0}^{T-1} \|u_t + w_t\|^2 \ge C  \right) \le \frac{1}{C} \bbE\left[ \sum_{t=0}^{T-1} \|u_t + w_t \|^2 \right] = T\frac{\bar{u}^2 + \sigma_w^2}{C} ,
	$
	where $ \bar{u} $ satisfies $ \|u\| \le \bar{u} $ for any $ u\in \bbU $. 
	Therefore, for any $ C > 0 $, it holds that
	$
		\bbP ( \regret_T \le \frac{1}{2\eta} D^2 + \frac{\eta G_c^2}{2} C ) \ge 1 - T\frac{\bar{u}^2 + \sigma_w^2}{C} .
	$
	Let $ C = T\frac{\bar{u}^2 + \sigma_w^2}{\delta} $ for $ \delta \in (0,1) $. 
	Then, by setting $ \eta = \sqrt{\delta/T} $, we obtain 
	\begin{align*}
		\bbP \left( \regret_T \le \frac{D^2 + G_c^2 (\bar{u}^2 + \sigma_w^2 )}{2\sqrt{\delta}} \sqrt{T}  \right) \ge 1 - \delta .
	\end{align*}

	Thus, even for this simplest case, $ 1/\sqrt{\delta} $ appears in the regret upper bound. By a similar argument, we can see that even for sub-Gaussian noise, $ 1/\sqrt{\delta} $ appears in the bound. 
    This is why an additional $ O(\log (1/\delta)) $ overhead for a regret bound that holds with probability $ 1-\delta $, would be impossible for online control under unbounded noise.

\newpage

\section{Notation}\label{app:notation}
\small
\begin{description}[leftmargin=3em, labelindent=1em]

\item[$n_{\rm x}$] State dimension.
\item[$n_{\rm u}$] Input dimension.
\item[$\dime = \max\{n_\rmx, n_\rmu\}$] 
\item[$\sigma_w$] Scale parameter of noise $w_t$ satisfying Assumption~\ref{ass:noise_4}.
\item[$G_c$] Constant satisfying Assumption~\ref{ass:cost}.
\item[$A_K = A - BK$] Closed-loop system matrix under state-feedback control $u_t = -K x_t$.
\item[$(\kappa,\gamma)$] Parameters for strong stability.
\item[$\kappa_B = \max\{\|B\|, 1\}$] 
\item[$K^*$] Optimal state-feedback gain in hindsight.
\item[${\sigma}_w^{[1,4]} = \max\{\sigma_w, \sigma_w^4\}$] 
\item[$D = \frac{4\kappa_B \kappa^3 \sqrt{\dime}}{\gamma}$] 
\item[$ \alpha $] Parameter of the strong convexity of $ c_t $.
\item[$ \beta $] Parameter of the smoothness of $ c_t $.
\item[$M_t = \{M_t^{[0]},\ldots,M_t^{[H-1]}\}$] Parameter of a disturbance-action policy.
\item[$\calM$] Set of admissible parameters of disturbance-action policies:
\[
\left\{
  M = \{M^{[0]},\ldots,M^{[H-1]}\} : 
  \|M^{[i]} \| \le 2\kappa_B\kappa^3 (1-\gamma)^i,\ \forall i\in \bbra{0,H-1}
\right\}
\]
\item[$F_t(M_{t-1-H:t})$] Surrogate cost in Definition~\ref{def:surrogate}:
\[
c_t (y_t^K(M_{t-1-H:t-1}), v_t^K(M_{t-1-H:t}) )
\]
\item[$f_t (M) = \surr_t (M,\ldots,M)$] 
\item[$\pi_t (K,M_t)$] Disturbance-action policy choosing the control input as
\begin{align*}
    u_t = - K x_t + \sum_{i=1}^H M_t^{[i-1]} w_{t-i} .
\end{align*}
\item[$x_t^K(M_{0:t-1})$] State $x_t$ driven by $\{\pi_s (K,M_s) \}_{s=0}^{t-1}$.
\item[$x_t^K(M)$] Shortcut notation for $x_t^K(M_{0:t-1})$ with $M_t \equiv M$.
\item[$x_t^K$] $x_t^K(M_{0:t-1})$ with $M_t \equiv 0$.
\item[$u_t^K(M_{0:t})$] Input $u_t$ given by $\{\pi_s (K,M_s) \}_{s=0}^{t}$.
\item[$u_t^K(M)$] Shortcut notation for $u_t^K(M_{0:t})$ with $M_t \equiv M$.
\item[$u_t^K$] $u_t^K(M_{0:t})$ with $M_t \equiv 0$.
\item[$y_t^K, v_t^K$] Surrogate state and input, i.e., $x_t$ starting from $x_{t-1-H} = 0$ driven by $\{\pi_s (K,M_{s}) \}_{s=t-1-H}^{t-1}$ and the corresponding input $u_t$.
\item[$\Psi_{t,i}^{K,h}(M_{t-h:t}) = \tilA_K^i  \bone_{i\le h} + \sum_{j=0}^h \tilA_K^j B M_{t-j}^{[i-j-1]} \bone_{i-j \in \bbra{1,H}}$] 
\item[$ \Pi_{\calM} (M_t - \eta \nabla f_t (M_t)) = \argmin_{M' \in \calM} \| M_t - \eta \nabla f_t (M_t) - M' \|_{\rm F} $] Orthogonal projection.
\item[$\lceil a \rceil$] Ceiling function: smallest integer greater than or equal to $a$.
\item[$\|\cdot\|$] Spectral norm or Euclidean norm.
\item[$\|\cdot\|_{\rm F}$] Frobenius norm.
\item[$A \succ 0$ / $A \succeq 0$] Matrix $A$ is positive definite / semidefinite.
\item[$A_{k:\ell}$] Sequence $\{A_k, A_{k+1}, \ldots, A_\ell\}$.
\item[$\bbra{k,\ell}$] Index set $\{k, k+1, \ldots, \ell\}$.
\item[$ \wtilde{\alpha} = \alpha \usigma^2 \gamma^2/(36 \kappa^{10}) $] 
\item[$ \usigma $] Parameter satisfying Assumption~\ref{ass:subgauss}.
\end{description}

\section{Proof of Lemma~\ref{lem:sufficiency}}\label{app:sufficiency}

Before giving the proof, we introduce a lemma.
\begin{lemma}[A modified version of Lemma 5.4 of \citealt{Agarwal2019}]
    Let $ K $ be $ (\kappa,\gamma) $-strongly stable for system~\eqref{eq:system} and $ M_t \in \calM $ for all $ t\in \bbZ_{\ge 0} $. Then, for any $ h,i\in \bbZ_{\ge 0} $ and $ t \ge h $, it holds that
    \begin{align}
        \|\Psi_{t,i}^{K,h} (M_{t-h:t}) \| &\le \kappa^2 (1-\gamma)^i \bone_{i\le h} + 2H\kappa_B^2 \kappa^5 (1-\gamma)^{i-1} . \nonumber
    \end{align}
    \hfill $ \diamondsuit $
\end{lemma}
It follows immediately from the above lemma that
\begin{align}
    \|\Psi_{t,i}^{K,h} (M_{t-h:t}) \| \le (2H+1)\kappa_B^2 \kappa^5 (1-\gamma)^{i-1}  , \label{eq:psi_inequality}
\end{align}
which we will use throughout the regret analysis in this paper.

First, by \eqref{eq:strong_stable_inequality} and the inequality $ \|XY\| \le \|X\|\|Y\| $ for any matrices $ X $ and $ Y $ of compatible dimensions, we have
\begin{align}
    \|M_*^{[i]} \| \le \| K - K^* \| \|A_{K^*}^i \|  \le 2\kappa^3 (1-\gamma)^i \le 2 \kappa_B \kappa^3 (1-\gamma)^i , ~~ \forall i \in \bbra{0,H-1} , \nonumber
\end{align}
which means $ M_* \in \calM $.

Next, we derive \eqref{eq:M*_bound}.
By the triangle inequality, we have
\begin{align}
    & \sum_{t=0}^{\ft-1} \left| c_t (x_t^K(M_*), u_t^K (M_*)) - c_t (x_t^{K^*} , u_t^{K^*} ) \right|  \nonumber\\
    &\le \underbrace{ \sum_{t=0}^{\ft-1} \left| c_t (x_t^K(M_*), u_t^K (M_*)) - c_t (x_t^{K^*} , u_t^K (M_*)) \right|  }_{=: \term_{1}} + \underbrace{ \sum_{t=0}^{\ft-1}\left| c_t (x_t^{K^*} , u_t^K (M_*)) - c_t (x_t^{K^*} , u_t^{K^*} ) \right| }_{=: \term_{2}} . \label{eq:L1_L2}
\end{align}
A differentiable convex function $ g : \bbR^n \rightarrow \bbR $ satisfies $ g(y) - g(x) \ge \nabla g(x)^\top (y-x) $ and $ g(x) - g(y) \ge \nabla g(y)^\top (x-y) $ for any $ x,y \in \bbR^n $. Hence, it holds that for any $ x,y \in \bbR^n $,
\begin{align*}
  |g(x) - g(y)| &\le \max \left\{ |\nabla g(y)^\top (x-y)|, |\nabla g(x)^\top (y-x)|  \right\} \nonumber\\
  &\le \max \left\{ \|\nabla g(x) \|, \| \nabla g(y) \| \right\} \| x- y \|  . 
\end{align*}
Using the above inequality, Assumption~\ref{ass:cost}, and H\"{o}lder's inequality, we obtain
\begin{align}
    \term_{1} &\le \sum_{t=0}^{\ft-1} \max \left\{ \| \nabla_x c_t(x_t^K(M_*), u_t^K(M_*)) \|, \| \nabla_x c_t(x_t^{K^*}, u_t^K(M_*)) \|  \right\} \| x_t^K (M_*) - x_t^{K^*} \| \nonumber\\
    &\le G_c \sum_{t=0}^{\ft-1} \max \left\{ \|x_t^K(M_*) \|, \| x_t^{K^*} \|   \right\} \| x_t^K (M_*) - x_t^{K^*} \| \nonumber\\
    &\le G_c \left( \sum_{t=0}^{\ft-1} \left(\| x_t^K (M_*) \| + \|x_t^{K^*} \| \right)  \right) \sum_{t=0}^{\ft-1}\| x_t^K (M_*) - x_t^{K^*} \| . \label{eq:term1_bound}
\end{align} 

In what follows, we will bound the components in \eqref{eq:term1_bound}.
By Markov's inequality, for any $ \thr_1 > 0 $,
\begin{align}
    \bbP \left( \sum_{t=0}^{\ft-1} \left( \| x_t^K(M_*) \| + \| x_t^{K^*} \| \right) \ge T \thr_1  \right) \le \frac{1}{T\thr_1} \bbE\left[ \sum_{t=0}^{\ft-1} \left( \| x_t^K(M_*) \| + \| x_t^{K^*} \| \right) \right] . \label{eq:markov_1}
\end{align}
As shown in the proof of Lemma~5.2 of \citet{Agarwal2019}, $ x_t^K (M_*) $ can be written as
\begin{align}
    x_{t}^K (M_*) = 
         \sum_{i=0}^H \tilA_{K^*}^i w_{t-1-i} + \sum_{i=H+1}^{t-1} \Psi_{t-1,i}^{K,t-1}(M_*) w_{t-1-i} ,  ~~ \forall t \ge 0 , \label{eq:x_M}
\end{align}
where $ w_t := 0, \forall t < 0 $ for notational simplicity.
When $ t\in \bbra{0,H+1} $, the third term is understood to be zero, and it follows from $ w_t = 0, \forall t < 0 $ that
\begin{align}
    x_t^{K} (M_*) =  \sum_{i=0}^{t-1} A_{K^*}^i w_{t-1-i} , ~~ \forall t \in \bbra{0,H+1} .\label{eq:x_M_small}
\end{align}
By \eqref{eq:x_M} and Assumption~\ref{ass:noise_4}-(i), for any $ t \ge 0 $, it holds that
\begin{align}
  \bbE\left[ \| x_t^K (M_*) \| \right] &= \bbE\left[ \left\|   \sum_{i=0}^H \tilA_{K^*}^i w_{t-1-i} + \sum_{i=H+1}^{t-1} \Psi_{t-1,i}^{K,t-1}(M_*) w_{t-1-i}       \right\|   \right] \nonumber\\
  &= \bbE\left[ \sum_{i=0}^H \|\tilA_{K^*}^i w_{t-1-i}\| + \sum_{i=H+1}^{t-1} \|\Psi_{t-1,i}^{K,t-1}(M_*) w_{t-1-i}\|    \right] \nonumber\\
  &= \bbE\left[  \sum_{i=0}^H \|\tilA_{K^*}^i \| \| w_{t-1-i}\| + \sum_{i=H+1}^{t-1} \|\Psi_{t-1,i}^{K,t-1}(M_*) \| \|  w_{t-1-i}\|    \right] \nonumber\\
  &\le   \sigma_w \left( \sum_{i=0}^H \kappa^2 (1-\gamma)^i + \sum_{i=H+1}^{t-1}  (2H+1)\kappa_B^2 \kappa^5 (1-\gamma)^{i-1}  \right)  , \label{eq:xM_bound}
\end{align}
where in the last line, we used \eqref{eq:strong_stable_inequality} and \eqref{eq:psi_inequality}.
Therefore, for $ t \ge H+2 $, we obtain
\begin{align}
    \bbE\left[ \| x_t^K (M_*) \| \right] &\le  \sigma_w \left(  \kappa^2 (1-\gamma)^H + \sum_{i=1}^{t-1}  (2H+1)\kappa_B^2 \kappa^5 (1-\gamma)^{i-1}  \right) \nonumber\\
    &\le \sigma_w \kappa^2 (1-\gamma)^H + \frac{\sigma_{w}  (2H+1)\kappa_B^2 \kappa^5}{\gamma} .\label{eq:xM2_bound}
\end{align}
Similarly, for $ t\in \bbra{0,H+1} $, \eqref{eq:x_M_small} yields that
\begin{align}
    \bbE\left[ \| x_t^K (M_*) \| \right] &\le  \sigma_w \sum_{i=0}^{t-1} \kappa^2 (1-\gamma)^{i} \le  \frac{\sigma_{w} \kappa^2}{\gamma} , \label{eq:xM2_bound_small}
\end{align}
which can be bounded from above by the right-hand side of \eqref{eq:xM2_bound}.

On the other hand, $ x_t^{K^*}$ following $ x_{t+1}^{K^*} = \tilA_{K^*} x_t^{K^*} + w_t $ \axv{for any $ t\ge 0 $} and $ x_0^{K^*} = 0 $ admits the expression:
\begin{align}
    x_{t}^{K^*} &=  \sum_{i=0}^{t-1} \tilA_{K^*}^i w_{t-1-i}  ,  ~~  \forall t \ge 1 . \label{eq:x_t^K_solution}
\end{align}
Then, similar to \eqref{eq:xM2_bound_small}, we have
\begin{align}
    \bbE\left[ \|x_t^{K^*} \|  \right] \le  \frac{\sigma_{w} \kappa^2}{\gamma}  , ~~ \forall t \ge 0 . \label{eq:x_bound}
\end{align}

By \eqref{eq:markov_1}, \eqref{eq:xM2_bound}, \eqref{eq:xM2_bound_small}, and \eqref{eq:x_bound}, for any $ \thr_1 > 0 $ and $ H, \ft \in \bbZ_{>0} $, we have
\begin{align}
    &\bbP \left( \sum_{t=0}^{\ft-1} \left( \| x_t^K(M_*) \| + \| x_t^{K^*} \| \right) \ge T \thr_1  \right) \nonumber\\
    &\le \frac{1}{T\thr_1}  \sum_{t=0}^{\ft-1} \left( \sigma_w \kappa^2 (1-\gamma)^H + \frac{\sigma_{w}  (2H+1)\kappa_B^2 \kappa^5}{\gamma}  + \frac{\sigma_{w} \kappa^2}{\gamma}  \right) \nonumber\\
    &\le \frac{1}{\thr_1} \left( \sigma_w \kappa^2 (1-\gamma)^H + \frac{2\sigma_{w}  (H+1)\kappa_B^2 \kappa^5}{\gamma} \right) . \label{eq:markov_12}
\end{align}
For $ H \ge 2\gamma^{-1} \log T $, it holds that $ (1-\gamma)^H T^2 \le 1 $. In fact, noting that $ \log (1-\gamma) \le -\gamma $, we have
    \begin{align*}
        H \ge 2\gamma^{-1} \log T &\ge - \frac{2\log T}{\log (1-\gamma)} = \log_{1-\gamma} \frac{1}{T^2} ,
    \end{align*}
which implies $ (1-\gamma)^H T^2 \le 1 $.
For any $ T \ge 3 > \Ee $ and $ \coe > 0 $, by letting $ H = \lceil 2\gamma^{-1} \log \ft  \rceil $ and $ \thr_1 = \coe \log \ft $ in \eqref{eq:markov_12}, we obtain
\begin{align}
    \bbP \left( \sum_{t=0}^{\ft-1} \left( \| x_t^K(M_*) \| + \| x_t^{K^*} \| \right) \ge \coe \ft \log \ft   \right)
    &\le \frac{1}{\coe} \left( \frac{\sigma_w \kappa^2}{\ft^2} + \frac{2\sigma_{w}  (2\gamma^{-1} + \frac{2}{\log \ft})\kappa_B^2 \kappa^5}{\gamma} \right) \nonumber\\
    &\le \frac{9\sigma_{w} \kappa_B^2 \kappa^5}{\coe\gamma^2}  . \label{eq:D1_bound1}
\end{align}

We next consider the last component in \eqref{eq:term1_bound}.
By Markov's inequality, for any $ \thr_2 > 0 $,
\begin{align}
    \bbP \left( \sum_{t=0}^{\ft-1} \| x_t^K(M_*) - x_t^{K^*} \| \ge \frac{\thr_2}{\ft}  \right) \le \frac{\ft}{\thr_2} \bbE\left[ \sum_{t=0}^{\ft-1} \| x_t^K(M_*) - x_t^{K^*} \| \right] . \label{eq:markov_2}
\end{align}
Moreover, by \eqref{eq:x_M},~\eqref{eq:x_M_small}, and \eqref{eq:x_t^K_solution}, we can rewrite $ \bbE[ \| x_{t}^K (M_*) - x_{t}^{K^*} \| ] $ as
\begin{align}
    \bbE\left[ \| x_{t}^K (M_*) - x_{t}^{K^*} \| \right] =
    \begin{cases}
     \bbE\left[ \left\|\sum_{i=H+1}^{t-1} \left(  \Psi_{t-1,i}^{K,t-1} (M_*) - \tilA_{K^*}^i \right) w_{t-1-i}  \right\|     \right], & t \ge H + 2, \\
      0, & t\in \bbra{0,H+1},
    \end{cases} \nonumber
\end{align}
which yields for any $ t \ge H+2 $,
\begin{align}
  \bbE\left[ \| x_{t}^K (M_*) - x_{t}^{K^*} \| \right] 
  &\le \bbE\left[  \sum_{i=H+1}^{t-1} \left(  \| \Psi_{t-1,i}^{K,t-1} (M_*) \| + \| \tilA_{K^*}^i \| \right) \|w_{t-1-i}\|      \right] \nonumber\\
  &\le   \sigma_w \sum_{i=H+1}^{t-1} \left((2H+1)\kappa_B^2 \kappa^5 (1-\gamma)^{i-1} + \kappa^2 (1-\gamma)^i \right)  \nonumber\\
  &\le  2\sigma_w \sum_{i=H+1}^{t-1} (H+1)\kappa_B^2 \kappa^5 (1-\gamma)^{i-1}   \nonumber\\
  &\le \frac{2\sigma_{w}(H+1)\kappa_B^2 \kappa^5 (1-\gamma)^H }{\gamma} . \label{eq:xM_x_bound}
\end{align} 
Therefore, by \eqref{eq:markov_2},
\begin{align}
    \bbP \left( \sum_{t=0}^{\ft-1} \| x_t^K(M_*) - x_t^{K^*} \| \ge \frac{\thr_2}{\ft}  \right) &\le  \frac{\ft}{\thr_2}  \sum_{t=H+2}^{\ft-1}\frac{2\sigma_{w}(H+1)\kappa_B^2 \kappa^5 (1-\gamma)^H }{\gamma} \nonumber\\
    &\le \frac{2\sigma_{w}(H+1)\kappa_B^2 \kappa^5 (1-\gamma)^H \ft^2 }{\thr_2 \gamma} . \nonumber
\end{align}
For any $ \ft \ge 3 $ and $ \coe > 0 $, by letting $ H = \lceil 2\gamma^{-1} \log \ft  \rceil $ and $ \thr_2 = \coe \log \ft $, we get
\begin{align}
    \bbP \left( \sum_{t=0}^{\ft-1} \| x_t^K(M_*) - x_t^{K^*} \| \ge \frac{\coe \log \ft}{\ft}  \right) 
    \le \frac{2\sigma_{w}(2\gamma^{-1} + \frac{2}{\log \ft})\kappa_B^2 \kappa^5 }{\coe \gamma}
    \le \frac{8\sigma_{w}\kappa_B^2 \kappa^5 }{\coe \gamma^2}. \label{eq:D1_bound2}
\end{align}
In summary, \eqref{eq:term1_bound}, \eqref{eq:D1_bound1}, and \eqref{eq:D1_bound2} yield the following bound of $ \term_1 $:
\begin{align}
    \bbP\left( \term_1 \le G_c \coe^2 (\log \ft)^2  \right) \ge 1 - \frac{17\sigma_{w} \kappa_B^2 \kappa^5}{\coe\gamma^2} , ~~ \forall \coe > 0, \ \ft \ge 3 . \label{eq:D1_bound3}
\end{align}

Similarly to $ \term_1 $, we next bound $ \term_2 $ in \eqref{eq:L1_L2} based on
\begin{align}
    \term_2 \le G_c \left( \sum_{t=0}^{\ft-1} \left( \|u_t^K(M_*) \| + \| u_t^{K^*} \| \right)  \right) \sum_{t=0}^{\ft-1} \| u_t^K(M_*) - u_t^{K^*} \| . \label{eq:D2_upper}
\end{align}
For $ t \in \bbra{0,H} $, by \eqref{eq:x_M_small},~\eqref{eq:x_t^K_solution}, and $ w_t = 0, \forall t < 0 $, the disturbance-action policy~\eqref{eq:disturbance-action} with $ M_* $ in \eqref{eq:M*} yields
\begin{align}
    u_t^K (M_*) &= - K \sum_{i=0}^{t-1} A_{K^*}^i w_{t-1-i} + (K-K^*) \sum_{i=0}^{t-1} A_{K^*}^i w_{t-1-i} \nonumber\\
    &= - K^* \sum_{i=0}^{t-1} A_{K^*}^i w_{t-1-i} \label{eq:uM_small_2}\\
    &= u_t^{K^*} . \label{eq:uM_small}
\end{align}
In view of this, we consider separately the two cases $ t \ge H+1 $ and $ t\in \bbra{0,H} $.

By \eqref{eq:x_M}, for $ t \ge H +1 $, it holds that
\begin{align}
    \bbE\left[ \|u_t^K (M_*)\| \right] &= \bbE\left[ \left\| -K x_t^K (M_*) + \sum_{i=1}^H M_*^{[i-1]} w_{t-i}  \right\| \right] \nonumber\\
    &\le \bbE\left[ \kappa \|x_t^K(M_*) \| + \sum_{i=1}^H 2\kappa_B \kappa^3 (1-\gamma)^{i-1} \|w_{t-i} \|  \right] \nonumber\\
    &\le \sigma_w \kappa^3 (1-\gamma)^H + \frac{\sigma_{w}  (2H+1)\kappa_B^2 \kappa^6}{\gamma} +  \frac{2\sigma_w \kappa_B\kappa^3}{\gamma}  \nonumber\\
    &\le \sigma_w \kappa^3 (1-\gamma)^{H} + \frac{ \sigma_{w} (2H+3)\kappa_B^2 \kappa^{6}}{\gamma} . \label{eq:uM_bound}
\end{align}

By \eqref{eq:uM_small_2}, for $ t \in \bbra{0,H} $, we have
\begin{align}
    \bbE\left[ \|u_t^K (M_*)\| \right] &\le  \bbE\left[\sum_{i=0}^{t-1} \|K^* \| \|A_{K^*}^i \| \|w_{t-1-i} \| \right] \nonumber\\
    &\le  \frac{\sigma_{w} \kappa^3 }{\gamma} \le \sigma_w \kappa^3 (1-\gamma)^{H} + \frac{ \sigma_{w} (2H+3)\kappa_B^2 \kappa^{6}}{\gamma}  . \label{eq:uM_moment_small}
\end{align}
Similarly, by \eqref{eq:x_t^K_solution},
\begin{align}
    \bbE\left[\|u_t^{K^*} \| \right] = \bbE\left[ \left\|  K^* \sum_{i=0}^{t-1} \tilA_{K^*}^i w_{t-1-i} \right\|  \right] \le \frac{\sigma_{w} \kappa^3}{\gamma}, ~~ \forall t \ge 0.  \label{eq:u_bound}
\end{align}
For the last component of \eqref{eq:D2_upper}, similar to \eqref{eq:xM_x_bound}, for $ t \ge H+1 $, we have
\begin{align}
    \bbE\left[ \| u_t^K (M_*) - u_t^{K^*} \| \right] &= \bbE\left[ \left\| K^* x_t^{K^*} - Kx_t^K (M_*) + \sum_{i=1}^H (K - K^*) \tilA_{K^*}^{i-1} w_{t-i} \right\| \right] \nonumber\\
    &= \bbE\left[ \left\|  (K^* - K)A_{K^*}^H w_{t-1-H} + \sum_{i=H+1}^{t-1} \left(K^*A_{K^*}^i - K \Psi_{t-1,i}^{K,t-1}(M_*) \right) w_{t-1-i}     \right\|  \right] \nonumber\\
    &\le  2\sigma_w \kappa^3 (1-\gamma)^{H} + \sigma_w \sum_{i=H+1}^{t-1} \left( \kappa^3 (1-\gamma)^i + (2H+1)\kappa_B^2 \kappa^6 (1-\gamma)^{i-1}  \right)        \nonumber\\
    &\le  2\sigma_w \kappa^3 (1-\gamma)^{H} + \sigma_w \frac{2(H+1)\kappa_B^2 \kappa^6 (1-\gamma)^H}{\gamma}    \nonumber\\
    &\le \frac{2\sigma_{w} (H+2) \kappa_B^2 \kappa^{6} (1-\gamma)^{H}}{\gamma} .
    \label{eq:uM_u_bound}
\end{align}
For $ t\in \bbra{0,H} $, it follows from \eqref{eq:uM_small} that
\begin{align}
    \bbE\left[ \| u_t^K (M_*) - u_t^{K^*} \| \right] &= 0 . \label{eq:u_diff_small}
\end{align}

Hence, by Markov's inequality and \eqref{eq:uM_bound}--\eqref{eq:u_diff_small}, for any $ \thr_3>0 $ and $ \thr_4 > 0 $, it holds that
\begin{align}
    \bbP \left( \sum_{t=0}^{\ft-1} \left( \| u_t^K(M_*) \| + \| u_t^{K^*} \| \right) \ge T \thr_3  \right) &\le \frac{1}{T\thr_3}  \sum_{t=0}^{\ft-1} \left( \sigma_w \kappa^3 (1-\gamma)^{H} + \frac{ \sigma_{w} (2H+3)\kappa_B^2 \kappa^{6}}{\gamma} + \frac{\sigma_{w} \kappa^3}{\gamma} \right) \nonumber\\
    &\le \frac{1}{\thr_3} \left( \sigma_w \kappa^3 (1-\gamma)^H + \frac{2\sigma_{w}  (H+2)\kappa_B^2 \kappa^6}{\gamma} \right) ,\nonumber\\
    \bbP \left( \sum_{t=0}^{\ft-1} \| u_t^K(M_*) - u_t^{K^*} \| \ge \frac{\thr_4}{\ft}  \right) &\le  \frac{\ft}{\thr_4}  \sum_{t=H+1}^{\ft-1}\frac{2\sigma_{w} (H+2) \kappa_B^2 \kappa^{6} (1-\gamma)^{H}}{\gamma} \nonumber\\
    &\le \frac{2\sigma_{w}(H+2)\kappa_B^2 \kappa^6 (1-\gamma)^H \ft^2 }{\thr_4 \gamma} . \nonumber
\end{align}
For any $ \ft \ge 3 $ and $ \coe > 0 $, by setting $ H = \lceil 2\gamma^{-1} \log \ft  \rceil $, $ \thr_3 = \thr_4 = \coe \log \ft $ in the above inequalities, we get
\begin{align}
    \bbP \left( \sum_{t=0}^{\ft-1} \left( \| u_t^K(M_*) \| + \| u_t^{K^*} \| \right) \ge \coe \ft \log \ft  \right) 
    &\le \frac{1}{\coe} \left(  \frac{\sigma_w \kappa^3}{\ft^2 \log \ft} + \frac{10\sigma_{w}  \kappa_B^2 \kappa^6}{\gamma^2} \right) \nonumber\\
    &\le \frac{11\sigma_{w}  \kappa_B^2 \kappa^6}{\coe \gamma^2},  \nonumber\\
    \bbP \left( \sum_{t=0}^{\ft-1} \| u_t^K(M_*) - u_t^{K^*} \| \ge \frac{\coe \log \ft}{\ft}  \right) 
    &\le \frac{10\sigma_{w}\kappa_B^2 \kappa^6  }{\coe  \gamma^2} . \nonumber
\end{align}
By combining the above with \eqref{eq:D2_upper}, we obtain
\begin{align}
    \bbP\left( \term_2 \le G_c \coe^2 (\log \ft)^2 \right) \ge 1 - \frac{21\sigma_{w}  \kappa_B^2 \kappa^6}{\coe \gamma^2} , ~~ \forall \coe > 0, \ \ft \ge 3 . \label{eq:D2_bound}
\end{align}
Lastly, by \eqref{eq:L1_L2},~\eqref{eq:D1_bound3}, and \eqref{eq:D2_bound}, we arrive at \eqref{eq:M*_bound}, which completes the proof.

\section{Proof of Lemma~\ref{thm:approximation}}\label{app:approximation}

As in the proof of Lemma~\ref{lem:sufficiency}, we derive \eqref{eq:surrogate_actual_bound} based on the following inequality:
\begin{align}
    &\sum_{t=0}^{\ft-1} \left| \surr_t (M_{t-1-H:t}) -c_t (x_t^K (M_{0:t-1}), u_t^K(M_{0:t}))      \right|  \nonumber\\
    &\le \underbrace{\sum_{t=0}^{\ft-1}\left| c_t (y_t^K(M_{t-1-H:t-1}), v_t^K (M_{t-1-H:t})) - c_t (x_t^K(M_{0:t-1}), v_t^K (M_{t-1-H:t})) \right|  }_{=: \term_{3}} \nonumber\\
    &\quad+ \underbrace{\sum_{t=0}^{\ft-1} \left| c_t (x_t^K(M_{0:t-1}), v_t^K (M_{t-1-H:t}))- c_t (x_t^K (M_{0:t-1}), u_t^K(M_{0:t})) \right| }_{=: \term_{4}} . \label{eq:diff_surr_bound}
\end{align}
Similar to \eqref{eq:term1_bound}, we have
\begin{align}
    \term_{3} &\le G_c \left(\sum_{t=0}^{\ft-1} \left( \|y_t^K(M_{t-1-H:t-1})\|  + \|x_t^K(M_{0:t-1})\| \right)  \right) \sum_{t=0}^{\ft-1} \|y_t^K(M_{t-1-H:t-1}) - x_t^K(M_{0:t-1}) \| , \label{eq:L3} \\
    \term_{4} &\le G_c \left(\sum_{t=0}^{\ft-1} \left( \|v_t^K(M_{t-1-H:t})\| +\|u_t^K(M_{0:t})\| \right)  \right) \sum_{t=0}^{\ft-1} \|v_t^K(M_{t-1-H:t}) - u_t^K(M_{0:t}) \| , \label{eq:L4}
\end{align}
and we will use Markov's inequality for bounding $ \term_3 $ and $ \term_4 $.

For \eqref{eq:L3}, by the expression of $ y_t^K(M_{t-1-H:t-1}) $ in \eqref{eq:surrogate_y} and the inequality~\eqref{eq:psi_inequality}, we obtain for any $ t \ge 0 $,
\begin{align}
    \bbE\left[ \|y_t^K (M_{t-1-H:t-1})\|  \right] &= \bbE\left[ \left\| \sum_{i=0}^{2H} \Psi_{t-1,i}^{K,H} (M_{t-1-H:t-1}) w_{t-1-i} \right\| \right] \nonumber\\
    &\le \sigma_w \sum_{i=0}^{2H} (2H+1) \kappa_B^2 \kappa^{5} (1-\gamma)^{i-1} \nonumber\\
    &\le \frac{ \sigma_w (2H+1) \kappa_B^2 \kappa^{5} }{\gamma(1-\gamma)} . \label{eq:y_bound}
\end{align}
By Proposition~\ref{prop:transfer}, $ x_t^K(M_{0:t-1}) $ in \eqref{eq:L3} can be written as
\begin{align}
    x_t^K(M_{0:t-1}) = \sum_{i=0}^{t-1} \Psi_{t-1,i}^{K,t-1} (M_{0:t-1}) w_{t-1-i}  , ~~ \forall t \ge 0 .\label{eq:xM_solution_general}
\end{align}
Therefore, we have for any $ t\ge 0 $,
\begin{align}
  \bbE\left[ \|x_t^K(M_{0:t-1})\|  \right] &= \bbE\left[ \left\|  \sum_{i=0}^{t-1} \Psi_{t-1,i}^{K,t-1} (M_{0:t-1}) w_{t-1-i}     \right\|  \right] \nonumber\\
  &\le   \sigma_w \sum_{i=0}^{t-1} (2H+1) \kappa_B^2 \kappa^{5} (1-\gamma)^{i-1} \nonumber\\
    &\le \frac{ \sigma_{w} (2H+1) \kappa_B^2 \kappa^{5} }{\gamma(1-\gamma)}, \label{eq:xM_bound_general}
\end{align}
which also yields for any $ t\ge 0 $,
\begin{align}
    \bbE\left[ \| y_t^K (M_{t-1-H:t-1}) - x_t^K(M_{0:t-1}) \|  \right] &= \bbE\left[ \|\tilA_K^{H+1} x_{t-1-H}^K (M_{0:t-2-H}) \|   \right] \nonumber\\
        &\le \frac{\sigma_{w} (2H+1) \kappa_B^2 \kappa^{7} (1-\gamma)^{H} }{\gamma} .\label{eq:y_x_bound}
\end{align}

By Markov's inequality, \eqref{eq:y_bound}, \eqref{eq:xM_bound_general}, and \eqref{eq:y_x_bound}, for any $ \thr_1 > 0 $ and $ \thr_2 > 0 $, it holds that
\begin{align}
    &\bbP\left(\sum_{t=0}^{\ft-1} \left( \|y_t^K(M_{t-1-H:t-1})\|  + \|x_t^K(M_{0:t-1})\| \right) \ge \ft \thr_1  \right) \nonumber\\
    &\le \frac{1}{\ft \thr_1} \bbE\left[\sum_{t=0}^{\ft-1} \left( \|y_t^K(M_{t-1-H:t-1})\|  + \|x_t^K(M_{0:t-1})\| \right) \right] \le  \frac{2 \sigma_w (2H+1) \kappa_B^2 \kappa^{5} }{\thr_1\gamma(1-\gamma)} , \nonumber\\
    &\bbP\left(\sum_{t=0}^{\ft-1} \|y_t^K(M_{t-1-H:t-1}) - x_t^K(M_{0:t-1}) \| \ge \frac{\thr_2}{\ft} \right) \nonumber\\
    &\le \frac{\ft}{\thr_2} \bbE\left[ \sum_{t=0}^{\ft-1} \| y_t^K (M_{t-1-H:t-1}) - x_t^K(M_{0:t-1}) \|  \right] \le \frac{\sigma_{w} (2H+1) \kappa_B^2 \kappa^{7} (1-\gamma)^{H}\ft^2 }{\thr_2 \gamma} . \nonumber
\end{align}
For any $ \ft \ge 3 $ and $ \coe > 0 $, by letting $ H = \lceil 2\gamma^{-1} \log \ft  \rceil $, $ \thr_1 = \thr_2 = \coe \log \ft $, we obtain
\begin{align}
    \bbP\left(\sum_{t=0}^{\ft-1} \left( \|y_t^K(M_{t-1-H:t-1})\|  + \|x_t^K(M_{0:t-1})\| \right) \ge \coe \ft \log \ft  \right) &\le  \frac{2 \sigma_w (\frac{4\log \ft}{\gamma}+3) \kappa_B^2 \kappa^{5} }{\coe \gamma(1-\gamma) \log \ft} \nonumber\\ 
    &\le  \frac{14 \sigma_w  \kappa_B^2 \kappa^{5} }{\coe \gamma^2(1-\gamma)} , \nonumber\\
    \bbP\left(\sum_{t=0}^{\ft-1} \|y_t^K(M_{t-1-H:t-1}) - x_t^K(M_{0:t-1}) \| \ge \frac{\coe \log \ft}{\ft} \right) &\le \frac{\sigma_{w} (\frac{4\log \ft}{\gamma}+3) \kappa_B^2 \kappa^{7} }{\coe \gamma \log \ft} \nonumber\\
    &\le \frac{7\sigma_{w}  \kappa_B^2 \kappa^{7}  }{\coe \gamma^2} . \nonumber
\end{align}
Therefore, it follows from \eqref{eq:L3} that
\begin{align}
    \bbP\left( \term_3 \le G_c \coe^2 (\log \ft)^2  \right) \ge 1 - \frac{21 \sigma_w  \kappa_B^2 \kappa^{7} }{\coe \gamma^2(1-\gamma)} . \label{eq:term3_bound}
\end{align}

Next, we bound $ \term_{4} $ in \eqref{eq:diff_surr_bound}. By the expression of $ v_t^K (M_{t-1-H:t}) $ in \eqref{eq:surrogate_v} and \eqref{eq:y_bound}, and by noting that $ M_t \in \calM $, we have for any $ t \ge 0 $,
\begin{align}
    \bbE\left[\|v_t^K(M_{t-H-1:t})\|  \right] &= \bbE\left[ \left\| -K y_t^K(M_{t-H-1:t-1}) + \sum_{i=1}^{H} M_t^{[i-1]} w_{t-i}    \right\|  \right] \nonumber\\
    &\le \frac{ \sigma_w (2H+1) \kappa_B^2 \kappa^{6} }{\gamma(1-\gamma)} +  2 \sigma_w \kappa_B \kappa^3 \sum_{i=1}^H (1-\gamma)^{i-1}  \nonumber\\
    &\le \frac{ \sigma_w (2H+1) \kappa_B^2 \kappa^{6} }{\gamma(1-\gamma)} +  \frac{2 \sigma_w \kappa_B \kappa^3}{\gamma} \nonumber\\
    &\le \frac{ \sigma_w (2H+3) \kappa_B^2 \kappa^{6} }{\gamma(1-\gamma)}. \label{eq:v_bound}
\end{align}
For the first moment of $ u_t^K(M_{0:t}) $, by \eqref{eq:xM_bound_general}, it holds that for any $ t\ge 0 $,
\begin{align}
  \bbE\left[ \|u_t^K(M_{0:t})\|  \right] &= \bbE\left[ \left\|  - K x_t^K(M_{0:t-1}) + \sum_{i=1}^{H} M_t^{[i-1]} w_{t-i}    \right\|  \right] \nonumber\\
  &\le \frac{ \sigma_{w} (2H+1) \kappa_B^2 \kappa^{6} }{\gamma(1-\gamma)} + \frac{2 \sigma_w \kappa_B \kappa^3}{\gamma} \nonumber\\
  &\le \frac{ \sigma_w (2H+3) \kappa_B^2 \kappa^{6} }{\gamma(1-\gamma)} .
\end{align}
For the last component of \eqref{eq:L4}, by \eqref{eq:y_x_bound}, we have for any $ t \ge 0 $,
\begin{align}
    \bbE\left[ \|v_t^K(M_{t-H-1:t}) - u_t^K(M_{0:t}) \| \right] &= \bbE\left[ \left\| -K \left(y_t^K (M_{t-H-1:t-1}) - x_t^K(M_{0:t-1})  \right)    \right\|  \right] \nonumber\\
    &\le \frac{\sigma_{w} (2H+1) \kappa_B^2 \kappa^{8} (1-\gamma)^{H} }{\gamma}  . \label{eq:v_u}
\end{align}

Thus, by Markov's inequality and \eqref{eq:v_bound}--\eqref{eq:v_u}, for any $ \thr_3 > 0 $ and $ \thr_4 > 0 $,
\begin{align}
    &\bbP\left(\sum_{t=0}^{\ft-1} \left( \|v_t^K(M_{t-1-H:t})\| +\|u_t^K(M_{0:t})\| \right) \ge \ft \thr_3  \right) \nonumber\\
    &\le \frac{1}{\ft \thr_3} \bbE\left[\sum_{t=0}^{\ft-1} \left( \|v_t^K(M_{t-1-H:t})\| +\|u_t^K(M_{0:t})\| \right) \right] \le  \frac{2\sigma_{w} (2H+3)\kappa_B^2 \kappa^{6}}{\thr_3\gamma(1-\gamma)} ,\nonumber\\
    &\bbP\left(\sum_{t=0}^{\ft-1} \|v_t^K(M_{t-1-H:t}) - u_t^K(M_{0:t}) \| \ge \frac{\thr_4}{\ft} \right) \nonumber\\
    &\le \frac{\ft}{\thr_4} \bbE\left[ \sum_{t=0}^{\ft-1} \|v_t^K(M_{t-1-H:t}) - u_t^K(M_{0:t}) \|  \right] \le \frac{\sigma_{w} (2H+1) \kappa_B^2 \kappa^{8} (1-\gamma)^{H} \ft^2}{\thr_4\gamma } .\nonumber
\end{align}
For any $ \ft \ge 3 $ and $ \coe > 0 $, by letting $ H = \lceil 2\gamma^{-1} \log \ft  \rceil $, $ \thr_3 = \thr_4 = \coe \log \ft $, we get
\begin{align}
    \bbP\left(\sum_{t=0}^{\ft-1} \left( \|v_t^K(M_{t-1-H:t})\| +\|u_t^K(M_{0:t})\| \right) \ge \coe \ft \log \ft  \right) &\le  \frac{2\sigma_{w} (\frac{4\log \ft}{\gamma}+5)\kappa_B^2 \kappa^{6}}{\gamma(1-\gamma)\coe \log \ft} \nonumber\\ 
    &\le \frac{18 \sigma_w  \kappa_B^2 \kappa^{6} }{\coe \gamma^2(1-\gamma)}   , \nonumber\\
    \bbP\left(\sum_{t=0}^{\ft-1} \|v_t^K(M_{t-1-H:t}) - u_t^K(M_{0:t}) \| \ge \frac{\coe \log \ft}{\ft} \right) &\le \frac{\sigma_{w} (\frac{4\log \ft}{\gamma}+3) \kappa_B^2 \kappa^{8} }{\gamma  \coe\log \ft} \nonumber\\
    &\le \frac{7\sigma_{w}  \kappa_B^2 \kappa^{8}  }{\coe \gamma^2} . \nonumber
\end{align}
Hence, by \eqref{eq:L4}, we obtain
\begin{align}
    \bbP\left( \term_4 \le G_c \coe^2 (\log \ft)^2 \right) \ge 1 - \frac{25 \sigma_w  \kappa_B^2 \kappa^{8} }{\coe \gamma^2(1-\gamma)} . \label{eq:term4_bound}
\end{align}
Consequently, by \eqref{eq:diff_surr_bound},~\eqref{eq:term3_bound}, and \eqref{eq:term4_bound}, we obtain the desired result \eqref{eq:surrogate_actual_bound}.

\section{Proof of Proposition~\ref{prop:OCO_regret}}\label{app:OCO}
By the standard argument for the regret analysis of OGD (e.g., \citealt{Hazan2016intro}), for any $ M\in \calM $, the sequence $ \{M_t\} $ satisfies
\begin{align}
    \sum_{t=0}^{\ft-1} f_t (M_t) - \sum_{t=0}^{\ft-1} f_t (M) \le \frac{1}{2\eta} \left(\sup_{M_1,M_2\in \calM} \| M_1 - M_2 \|_{\rm F} \right)^2 + \frac{\eta}{2} \sum_{t=0}^{\ft-1} \| \nabla_M f_t (M_t) \|_{\rm F}^2  . \nonumber
\end{align}
Note that for any $ M \in \calM $,
\begin{align}
    \|M\|_{\rm F} \le \sum_{i=0}^{H-1} \|M^{[i]} \|_{\rm F} \le \sqrt{\dime} \sum_{i=0}^{H-1} \|M^{[i]} \|  \le \sqrt{\dime} \sum_{i=0}^{H-1} 2\kappa_B \kappa^3 (1-\gamma)^i \le \frac{2\kappa_B \kappa^3 \sqrt{\dime}}{\gamma} . \nonumber
\end{align}
Thus, by the triangle inequality, we obtain
\begin{align}
    \sum_{t=0}^{\ft-1} f_t (M_t) - \sum_{t=0}^{\ft-1} f_t (M) \le \frac{1}{2\eta} D^2 + \frac{\eta}{2} \sum_{t=0}^{\ft-1} \| \nabla_M f_t (M_t) \|_{\rm F}^2  . \label{eq:standard_ogd}
\end{align}
By the coordinate-wise Lipschitzness of $ \surr_t$ and the triangle inequality,
we also have for any $ t \ge H + 1 $,
\begin{align}\label{eq:ocom2}
    \left| \surr_t(M_{t-1-H}, \dots, M_t) 
    - f_t(M_t) \right|
    &\leq 
    L_\rmc \sum_{i=1}^{H+1}
    \|
        M_t - M_{t - i}
    \|_{\rm F} 
    \leq 
    L_\rmc \sum_{i=1}^{H+1} \sum_{k=1}^i
    \|
        M_{t - k + 1} - M_{t - k}
    \|_{\rm F}
    .
\end{align}
The right-hand side of~\eqref{eq:ocom2} is bounded as
\begin{equation}\label{eq:ocom3}
    L_\rmc \sum_{i=1}^{H+1} \sum_{k=1}^i
    \|
        M_{t - k + 1} - M_{t - k}
    \|_{\rm F}
    \leq 
    L_\rmc \eta \sum_{i=1}^{H+1} \sum_{k=1}^i
    \|
        \nabla f_{t-k}(M_{t-k})
    \|_{\rm F}
    .
\end{equation}
Here, the above inequality \eqref{eq:ocom3} follows since
letting $ \tilde{M}_{t+1} := \argmin_{M' \in (\bbR^{\dime_\rmu \times \dime_\rmx })^H} \| M_t - \eta \nabla f_t (M_t) - M' \|_{\rm F} $ 
and using the OGD update rule in~\eqref{eq:ogd} that holds for $t \in \bbra{0,T-1}$,
we have 
$
\| M_{t - k + 1} - M_{t - k} \|_{\rm F} 
\leq 
\| \tilde{M}_{t - k + 1} - M_{t - k} \|_{\rm F}
\leq
\eta
\|
    \nabla f_{t-k}(M_{t-k})
\|_{\rm F}
$.

By the definition of $ \surr_t $ and $ w_s = 0 $ for any $ s \le -1 $, $ \surr_t $ does not depend on $ \{M_s\}_{s\le 0} $ when $ t < H + 1 $.
Therefore, for $ t < H + 1 $,
\begin{align}\label{eq:ocom2_2}
    &| \surr_t(M_{t-1-H}, \dots, M_{-1},M_0, \ldots, M_t) 
    - f_t(M_t) | = | \surr_t(M_t, \dots, M_t, M_0, \ldots, M_t) 
    - f_t(M_t) | \nonumber\\
    &\leq 
    L_\rmc \sum_{i=1}^{t}
    \|
        M_t - M_{t - i}
    \|_{\rm F}
    \leq 
    L_\rmc \sum_{i=1}^{t} \sum_{k=1}^i
    \|
        M_{t - k + 1} - M_{t - k}
    \|_{\rm F}
    \leq 
    L_\rmc \eta \sum_{i=1}^{t} \sum_{k=1}^i
    \|
        \nabla f_{t-k}(M_{t-k})
    \|_{\rm F}
    .
\end{align}
By \eqref{eq:ocom2}--\eqref{eq:ocom2_2}, for any $ t \in \bbra{0,\ft-1} $, it holds that
\begin{align}
    &| \surr_t(M_{t-1-H}, \dots, M_t) 
    - f_t(M_t) | 
    \leq 
    L_\rmc \eta \sum_{i=1}^{\min\{H+1,t\}} \sum_{k=1}^i
    \|
        \nabla f_{t-k}(M_{t-k})
    \|_{\rm F}
    . \nonumber
\end{align}
By combining this with \eqref{eq:standard_ogd} and noting that
\[
 \sum_{t=0}^{\ft-1} \surr_t (M_{t-1-H:t}) - \sum_{t=0}^{\ft-1} f_t(M) \le \sum_{t=0}^{\ft-1} |\surr_t (M_{t-1-H:t}) -  f_t(M_t)| + \sum_{t=0}^{\ft-1} f_t(M_t) - \sum_{t=0}^{\ft-1} f_t(M),
\]
we obtain the desired result~\eqref{eq:general_ogd_bound}.

\section{Lipschitz Continuity of Surrogate Cost}\label{app:lipshitz}

\subsection{Proof of Lemma~\ref{lem:lipshitz}}\label{app:lipschitz_coordinate}
For notational simplicity, we denote $ y_{t}^K (M_{t-1-H:t-1}) $ and $ y_t^K (M_{t-1-H:t-k-1},\wcheck{M}_{t-k},M_{t-k+1:t-1}) $ by $ y_{t}^K $ and $ \wcheck{y}_t^K $, respectively. We use a similar notation for $ v_{t}^K (M_{t-1-H:t}) $ as $ v_t^K $ and $\wcheck{v}_t^K $. Then, we have
\begin{align}
  &\left| \surr_t (M_{t-1-H:t}) -  \surr_t (M_{t-1-H},\ldots,\wcheck{M}_{t-k},\ldots,M_{t}) \right| \nonumber\\
  &\qquad \qquad \le  \left| c_t(y_t^K,v_t^K) - c_t(\wcheck{y}_t^K,v_t^K) \right| + \left|c_t(\wcheck{y}_t^K,v_t^K) - c_t(\wcheck{y}_t^K,\wcheck{v}_t^K) \right| . \label{eq:f_ftilde}
\end{align}
For the first term of the right-hand side of \eqref{eq:f_ftilde}, it holds that
\begin{align}
   \left| c_t(y_t^K,v_t^K) - c_t(\wcheck{y}_t^K,v_t^K) \right| &\le G_c\left( \|y_t^K \| +  \|\wcheck{y}_t^K \| \right) \|y_t^K - \wcheck{y}_t^K \| . \label{eq:lipschitz_1}
\end{align}
By Markov's inequality and \eqref{eq:y_bound}, for any $ \thr_1 > 0 $,
\begin{align}
    \bbP\left( \|y_t^K \| + \|\wcheck{y}_t^K \| \ge \thr_1  \right) &\le \frac{1}{\thr_1} \bbE \left[ \|y_t^K \| + \|\wcheck{y}_t^K \| \right] \nonumber\\
    &\le \frac{ 2\sigma_w (2H+1) \kappa_B^2 \kappa^{5} }{\thr_1 \gamma(1-\gamma)} . \label{eq:markov_y}
\end{align}
For $ k = 0 $, \axv{it holds that} $ y_t^K - \wcheck{y}_t^K = 0 $, and for $ k\in \bbra{1,H+1} $, using \eqref{eq:psi_def},~\eqref{eq:surrogate_y}, we get for any $ t \ge 0 $,
\begin{align}
    \bbE\left[ \| y_t^K - \wcheck{y}_t^K \|  \right] &= \bbE\left[ \left\| \tilA_K^{k-1} B \sum_{i=0}^{2H} (M_{t-k}^{[i-k]} - \wcheck{M}_{t-k}^{[i-k]}) w_{t-1-i} \bone_{i-k+1\in \bbra{1,H}}         \right\|  \right] \nonumber\\
    &\le \bbE\left[  \kappa_B \kappa^2 (1-\gamma)^{k-1}  \sum_{i=0}^{2H} \| M_{t-k}^{[i-k]} - \wcheck{M}_{t-k}^{[i-k]} \| \|w_{t-1-i}\| \bone_{i-k+1\in \bbra{1,H}}    \right] \nonumber         \\
    &\le \sigma_w \kappa_B \kappa^2 (1-\gamma)^{k-1} \sum_{i=1}^H \| M_{t-k}^{[i-1]} - \wcheck{M}_{t-k}^{[i-1]} \|  . \label{eq:y_ycheck_diff} 
\end{align}
Therefore, for any $ \thr_2 > 0 $ and $ M_{t-k}, \wcheck{M}_{t-k} \in \calM $, $ M_{t-k} \neq \wcheck{M}_{t-k}$, $ k \in \bbra{1,H+1} $,
\begin{align}
    \bbP\left( \|y_t^K - \wcheck{y}_t^K \| \ge \thr_2 \sum_{i=1}^H \| M_{t-k}^{[i-1]} - \wcheck{M}_{t-k}^{[i-1]} \| \right) &\le \frac{\bbE\left[ \|y_t^K - \wcheck{y}_t^K \|  \right] }{\thr_2 \sum_{i=1}^H \| M_{t-k}^{[i-1]} - \wcheck{M}_{t-k}^{[i-1]} \|}  \nonumber\\
    &\le \frac{\sigma_w \kappa_B \kappa^2 (1-\gamma)^{k-1}}{\thr_2} \le \frac{\sigma_w \kappa_B \kappa^2}{\thr_2} . \label{eq:cheby_y}
\end{align}
For any $ \coe > 0 $, $ \ft \ge 3 $, and $ M_{t-k}, \wcheck{M}_{t-k} \in \calM $, $ M_{t-k} \neq \wcheck{M}_{t-k}$, $ k \in \bbra{1,H+1} $, by setting $ H = \lceil 2\gamma^{-1} \log \ft  \rceil $, $ \thr_1 = \coe \log \ft $, and $ \thr_2 = \coe $ in \eqref{eq:lipschitz_1},~\eqref{eq:markov_y}, and \eqref{eq:cheby_y}, we obtain
\begin{align}
    &\bbP\left( \left| c_t(y_t^K,v_t^K) - c_t(\wcheck{y}_t^K,v_t^K) \right| \le G_c \coe^2 \log \ft \sum_{i=1}^H \| M_{t-k}^{[i-1]} - \wcheck{M}_{t-k}^{[i-1]} \|  \right) \nonumber\\
    &\ge 1 - \frac{ 2\sigma_w (2\gamma^{-1} \log \ft+3) \kappa_B^2 \kappa^{5} }{\coe \gamma(1-\gamma) \log \ft }  - \frac{\sigma_w \kappa_B \kappa^2}{\coe} \nonumber\\
    &\ge 1 - \frac{ 10\sigma_w  \kappa_B^2 \kappa^{5} }{\coe \gamma^2(1-\gamma) }  - \frac{\sigma_w \kappa_B \kappa^2}{\coe} \nonumber\\
    &\ge 1 - \frac{ 11\sigma_w  \kappa_B^2 \kappa^{5} }{\coe \gamma^2(1-\gamma) } . \label{eq:prob_y}
\end{align}

Next, for the second term of \eqref{eq:f_ftilde}, we have
\begin{align}
 \left| c_t (\wcheck{y}_t^K , v_t^K) - c_t (\wcheck{y}_t^K , \wcheck{v}_t^K) \right|  &\le G_c \left(  \|v_t^K\| + \|\wcheck{v}_t^K\| \right)  \|v_t^K - \wcheck{v}_t^K \| . \label{eq:lipschitz_2}
\end{align}
In addition, for any $ \thr_3 > 0 $, we have
\begin{align}
    \bbP\left( \|v_t^K \| + \|\wcheck{v}_t^K \| \ge \thr_3  \right) &\le \frac{1}{\thr_3} \bbE \left[ \|v_t^K \| + \|\wcheck{v}_t^K \| \right] \nonumber\\
    &\le \frac{2\sigma_w (2H+3)\kappa_B^2 \kappa^{6}}{\thr_3 \gamma (1-\gamma)} , \label{eq:markov_v_v}
\end{align}
where we used \eqref{eq:v_bound}.
By \eqref{eq:surrogate_v} and \eqref{eq:y_ycheck_diff}, the last component of \eqref{eq:lipschitz_2} satisfies that for any $ k \in \bbra{1,H+1} $ and $ t \ge 0 $,
\begin{align}
    \bbE\left[ \|v_t^K - \wcheck{v}_t^K \|  \right] &= \bbE\left[ \left\| -K(y_t^K - \wcheck{y}_t^K)  \right\| \right] \nonumber\\
    &\le  \sigma_w\kappa_B \kappa^3 (1-\gamma)^{k-1} \sum_{i=1}^H \| M_{t-k}^{[i-1]} - \wcheck{M}_{t-k}^{[i-1]} \|  .
  \end{align}
For $ k = 0 $ and for any $ t \ge 0 $,
\begin{align}
    \bbE\left[ \|v_t^K - \wcheck{v}_t^K \|  \right] &= \bbE\left[ \left\| \sum_{i=1}^H (M_t^{[i-1]} - \wcheck{M}_t^{[i-1]} ) w_{t-i} \right\| \right] \nonumber\\
    &\le \sigma_w \sum_{i=1}^H \|M_t^{[i-1]} - \wcheck{M}_t^{[i-1]} \| .
\end{align}
Hence, for any $ \thr_4 > 0 $ and $ M_{t-k}, \wcheck{M}_{t-k} \in \calM $, $ M_{t-k} \neq \wcheck{M}_{t-k}$, $ k \in \bbra{1,H+1} $,
\begin{align}
    \bbP\left( \|v_t^K - \wcheck{v}_t^K \| \ge \thr_4 \sum_{i=1}^H \| M_{t-k}^{[i-1]} - \wcheck{M}_{t-k}^{[i-1]} \| \right) &\le \frac{\bbE\left[ \|v_t^K - \wcheck{v}_t^K \|  \right]}{\thr_4\sum_{i=1}^H \| M_{t-k}^{[i-1]} - \wcheck{M}_{t-k}^{[i-1]} \|} \nonumber\\
    &\le \frac{\sigma_w\kappa_B \kappa^3}{\thr_4} . \label{eq:chebyshev_v}
\end{align}
For any $ \coe > 0 $, $ \ft \ge 3 $, and $ M_{t-k}, \wcheck{M}_{t-k} \in \calM $, $ M_{t-k} \neq \wcheck{M}_{t-k}$, $ k \in \bbra{1,H+1} $, by setting $ H = \lceil 2\gamma^{-1} \log \ft  \rceil $, $ \thr_3 = \coe \log \ft $, and $ \thr_4 = \coe $ in \eqref{eq:lipschitz_2},~\eqref{eq:markov_v_v}, and \eqref{eq:chebyshev_v}, we have
\begin{align}
    &\bbP\left( \left| c_t (\wcheck{y}_t^K , v_t^K) - c_t (\wcheck{y}_t^K , \wcheck{v}_t^K) \right| \le G_c \coe^2 \log \ft \sum_{i=1}^H \| M_{t-k}^{[i-1]} - \wcheck{M}_{t-k}^{[i-1]} \|  \right) \nonumber\\
    &\ge 1 - \frac{2\sigma_w (4\gamma^{-1}\log \ft +5)\kappa_B^2 \kappa^{6}}{\coe \gamma (1-\gamma) \log \ft} - \frac{\sigma_w\kappa_B \kappa^3}{\coe} \nonumber\\
    &\ge 1 - \frac{19\sigma_w \kappa_B^2 \kappa^{6}}{\coe \gamma^2 (1-\gamma)} . \label{eq:prob_v}
\end{align}

In summary, by \eqref{eq:f_ftilde},~\eqref{eq:prob_y}, and \eqref{eq:prob_v}, for any $ C > 0 $, $ \ft \ge 3 $, and $ M_{t-k}, \wcheck{M}_{t-k} \in \calM $, $ M_{t-k} \neq \wcheck{M}_{t-k}$, $ k \in \bbra{1,H+1} $, we obtain
\begin{align}
    &\bbP\left( \left| \surr_t (M_{t-1-H:t}) -  \surr_t (M_{t-1-H},\ldots,\wcheck{M}_{t-k},\ldots,M_{t}) \right| \le 2 G_c \coe^2 \log \ft \sum_{i=1}^H \| M_{t-k}^{[i-1]} - \wcheck{M}_{t-k}^{[i-1]} \| \right) \nonumber\\
    &\ge  1 - \frac{ 11\sigma_w  \kappa_B^2 \kappa^{5} }{\coe \gamma^2(1-\gamma) } - \frac{19\sigma_w \kappa_B^2 \kappa^{6}}{\coe \gamma^2 (1-\gamma)} \nonumber\\
    &\ge 1 - \frac{30\sigma_w \kappa_B^2 \kappa^{6}}{\coe \gamma^2 (1-\gamma)} . \nonumber
\end{align}
Lastly, we have
\begin{align}
    \sum_{i=1}^H \|M_{t-k}^{[i-1]} - \wcheck{M}_{t-k}^{[i-1]} \| &\le \sqrt{H} \left(\sum_{i=1}^H \|M_{t-k}^{[i-1]} - \wcheck{M}_{t-k}^{[i-1]} \|^2 \right)^{1/2} \nonumber\\
    &\le \sqrt{H} \left\|
  \begin{bmatrix}
    M_{t-k}^{[0]} - \wcheck{M}_{t-k}^{[0]}  \\
    \vdots \\
    M_{t-k}^{[H-1]} - \wcheck{M}_{t-k}^{[H-1]}  
  \end{bmatrix} \right\|_{\rm F} . \label{eq:1_Fnorm}
\end{align}
This completes the proof.

\subsection{Proof of Lemma~\ref{lem:lipschitz_global}}
For simplicity, we will omit the superscript of $ y_{t}^K$ and $v_t^K $ as $ y_t(M) $ and $ v_t(M) $.
By the same argument as in Appendix~\ref{app:lipschitz_coordinate}, we have for any $ t \ge 0 $, $ \thr_1 > 0 $, and $ \thr_3 > 0 $,
\begin{align}
  \left| f_t (M) -  f_t (\wcheck{M}) \right| &\le  G_c\left( \|y_t(M) \| +  \|y_t(\wcheck{M}) \| \right) \|y_t(M) - y_t(\wcheck{M}) \| \nonumber\\
  &\quad + G_c \left(  \|v_t(M)\| + \|v_t(\wcheck{M})\| \right)  \|v_t(M) - v_t(\wcheck{M}) \| , \label{eq:f_ftilde_global} \\
  \bbP\left( \|y_t(M) \| + \|y_t(\wcheck{M}) \| \ge \thr_1  \right) &\le \frac{ 2\sigma_w (2H+1) \kappa_B^2 \kappa^{5} }{\thr_1 \gamma(1-\gamma)} , \label{eq:markov_y_global} \\
  \bbP\left( \|v_t(M) \| + \|v_t(\wcheck{M}) \| \ge \thr_3  \right) 
    &\le \frac{2\sigma_w (2H+3)\kappa_B^2 \kappa^{6}}{\thr_3 \gamma (1-\gamma)}  . \label{eq:markov_v}
\end{align}
For the first term of \eqref{eq:f_ftilde_global}, we have for any $ t \ge 0 $,
\begin{align}
    \bbE\left[ \|y_t(M) - y_t(\wcheck{M}) \| \right] &= \bbE\left[ \left\|    \sum_{i=0}^{2H} \sum_{j=0}^H A_K^j B(M^{[i-j-1]} - \wcheck{M}^{[i-j-1]})\bone_{i-j\in \bbra{1,H}}   w_{t-1-i}  \right\|  \right] \nonumber\\
    &\le \sigma_w \kappa_B \kappa^2 \sum_{i=0}^{2H} \sum_{j=0}^{H} \|M^{[i-j-1]} - \wcheck{M}^{[i-j-1]}\|\bone_{i-j\in \bbra{1,H}} \nonumber\\
    &\le \sigma_w \kappa_B \kappa^2 \sum_{j=0}^{H} \sum_{i=1}^{H} \|M^{[i-1]} - \wcheck{M}^{[i-1]}\| \nonumber\\
    &= \sigma_w \kappa_B \kappa^2 (H+1) \sum_{i=1}^{H} \|M^{[i-1]} - \wcheck{M}^{[i-1]}\|  .\nonumber
\end{align}
For the second term of \eqref{eq:f_ftilde_global}, by \eqref{eq:surrogate_v}, it holds that for any $ t \ge 0 $,
\begin{align}
    \bbE\left[ \|v_t(M) - v_t(\wcheck{M}) \|  \right] &= \bbE\left[ \left\| -K(y_t(M) - y_t(\wcheck{M})) + \sum_{i=1}^H (M^{[i-1]} - \wcheck{M}^{[i-1]}) w_{t-i}  \right\| \right] \nonumber\\
    &\le \sigma_w \kappa_B \kappa^3 (H+1) \sum_{i=1}^{H} \|M^{[i-1]} - \wcheck{M}^{[i-1]}\| + \sigma_w \sum_{i=1}^H \|M^{[i-1]} - \wcheck{M}^{[i-1]}\|    \nonumber\\
    &\le \sigma_w \kappa_B \kappa^3 (H+2) \sum_{i=1}^{H} \|M^{[i-1]} - \wcheck{M}^{[i-1]}\|   .\nonumber
  \end{align}
Hence, for any $ \thr_2 > 0 $, $\thr_4 > 0 $, and $ M, \wcheck{M} \in \calM $, $ M \neq \wcheck{M}$,
\begin{align}
    \bbP\left( \|y_t(M) - y_t(\wcheck{M}) \| \ge \thr_2 \sum_{i=1}^H \| M^{[i-1]} - \wcheck{M}^{[i-1]} \| \right) &\le \frac{\bbE\left[ \|y_t(M) - y_t(\wcheck{M}) \|  \right] }{\thr_2 \sum_{i=1}^H \| M^{[i-1]} - \wcheck{M}^{[i-1]} \|}  \nonumber\\
    &\le \frac{\sigma_w \kappa_B \kappa^2(H+1)}{\thr_2} , \label{eq:cheby_y_global} \\
    \bbP\left( \|v_t(M) - v_t(\wcheck{M}) \| \ge \thr_4 \sum_{i=1}^H \| M^{[i-1]} - \wcheck{M}^{[i-1]} \| \right) &\le \frac{\bbE\left[ \|v_t(M) - \wcheck{v}_t(\wcheck{M}) \|  \right]}{\thr_4\sum_{i=1}^H \| M^{[i-1]} - \wcheck{M}^{[i-1]} \|} \nonumber\\
    &\le \frac{\sigma_w\kappa_B \kappa^3(H+2)}{\thr_4} . \label{eq:chebyshev_v_global}
\end{align}
For any $ \ft \ge 3 $, by setting $ H = \lceil 2\gamma^{-1} \log \ft  \rceil $, $ \thr_1 = \thr_2 = \thr_3 = \thr_4 = \coe \log \ft $ and by \eqref{eq:f_ftilde_global}--\eqref{eq:chebyshev_v_global}, we obtain
\begin{align}
    &\bbP\left( \left| f_t (M) -  f_t (\wcheck{M}) \right| \le 2 G_c \coe^2 (\log \ft)^2 \sum_{i=1}^H \| M^{[i-1]} - \wcheck{M}^{[i-1]} \| \right) \nonumber\\
    &\ge 1 -\frac{ 14\sigma_w  \kappa_B^2 \kappa^{5} }{\coe\gamma^2(1-\gamma)} - \frac{18\sigma_w \kappa_B^2 \kappa^{6}}{\coe \gamma^2 (1-\gamma)} - \frac{4\sigma_w \kappa_B \kappa^2}{\coe \gamma} - \frac{5\sigma_w\kappa_B \kappa^3}{\coe \gamma} \nonumber\\
    &\ge 1 - \frac{41\sigma_w \kappa_B^2 \kappa^{6}}{\coe \gamma^2 (1-\gamma)} .
\end{align}
By combining this with \eqref{eq:1_Fnorm}, we obtain the desired result.

\section{Proof of Lemma~\ref{lem:gradient}}\label{app:gradient}

For $ r\in \bbra{0,H-1} $, $p\in \bbra{1,n_\rmu}$, and $q\in \bbra{1,n_\rmx} $, let us bound $ \nabla_{M_{p,q}^{[r]}} f_t (M) $, where $ M_{p,q}^{[r]} $ denotes the $ (p,q) $ entry of $ M^{[r]} $.
In this section, for notational simplicity, we drop the superscript of $ y_t^K $ and $v_t^K $.
By using the chain rule, we obtain
\begin{align}
  \left|\nabla_{M_{p,q}^{[r]}} f_t (M) \right| &= \left| \nabla_x c_t (y_t(M),v_t(M))^\top \frac{\partial y_t(M)}{\partial M_{p,q}^{[r]}} + \nabla_u c_t (y_t (M), v_t(M))^\top \frac{\partial v_t (M)}{\partial M_{p,q}^{[r]}} \right| \nonumber\\
  &\le \|\nabla_x c_t (y_t(M),v_t(M))\| \left\|\frac{\partial y_t(M)}{\partial M_{p,q}^{[r]}}\right\| + \|\nabla_u c_t (y_t (M), v_t(M))\| \left\| \frac{\partial v_t (M)}{\partial M_{p,q}^{[r]}} \right\| \nonumber\\
  &\le G_c  \|y_t (M)\|   \left\|\frac{\partial y_t(M)}{\partial M_{p,q}^{[r]}}\right\|  + G_c  \| v_t(M) \|  \left\| \frac{\partial v_t (M)}{\partial M_{p,q}^{[r]}} \right\|  . \label{eq:grad_pq_bound}
\end{align}
By the expression \eqref{eq:surrogate_y} of $ y_t(M) $ and the bound \eqref{eq:psi_inequality} of $ \Psi $, for any $ M \in \calM $, we have
\begin{align}
    \|y_t(M) \| &=  \left\| \sum_{i=0}^{2H} \Psi_{t-1,i}^{K,H} (M) w_{t-1-i} \right\| \le \sum_{i=0}^{2H} (2H+1) \kappa_B^2 \kappa^5 (1-\gamma)^{i-1} \|w_{t-1-i}\| , \\
    \left\| \frac{\partial y_t(M)}{\partial M_{p,q}^{[r]}} \right\| &= \left\| \sum_{i=0}^{2H} \sum_{j=0}^H \frac{\partial \tilA_K^j B M^{[i-j-1]}}{\partial M_{p,q}^{[r]}} w_{t-1-i} \bone_{i-j\in \bbra{1,H}}          \right\| \nonumber\\
    &= \left\| \sum_{i=r+1}^{r+H+1} \frac{\partial \tilA_K^{i-r-1} BM^{[r]}}{\partial M_{p,q}^{[r]}} w_{t-1-i}     \right\| \nonumber\\
    &\le \sum_{i=r+1}^{r+H+1} \| \tilA_{K}^{i-r-1} B E^{p,q}     \| \|w_{t-1-i} \| \nonumber\\
    &\le \sum_{i=r+1}^{r+H+1} \kappa_B \kappa^2 (1-\gamma)^{i-r-1} \|w_{t-1-i} \|, \label{eq:dydM}
\end{align}
where the square matrix $ E^{p,q} $ is defined by $ E^{p,q}_{i,j} = 1 $ if $ (i,j) = (p,q) $, and $ 0 $, otherwise, and satisfies $ \|E^{p,q} \| = 1 $.
Similarly, we have for any $ M\in \calM $,
\begin{align}
    \|v_t(M) \| &= \left\| -K y_t(M) + \sum_{i=0}^{H-1} M^{[i]} w_{t-1-i}    \right\|  \nonumber\\
    &\le \sum_{i=0}^{2H} (2H+1) \kappa_B^2 \kappa^6 (1-\gamma)^{i-1} \|w_{t-1-i}\| + \sum_{i=0}^{H-1} 2\kappa_B \kappa^3 (1-\gamma)^i \|w_{t-1-i}\| \nonumber\\
    &\le \sum_{i=0}^{2H}  (2H+3)\kappa_B^2 \kappa^{6} (1-\gamma)^{i-1} \|w_{t-1-i}\|    ,    \\ 
    \left\| \frac{\partial v_t(M)}{\partial M_{p,q}^{[r]}}   \right\| &= \left\| -\frac{\partial K y_t (M)}{\partial M_{p,q}^{[r]}} + \sum_{i=0}^{H-1} \frac{\partial M^{[i]}}{\partial M_{p,q}^{[r]}} w_{t-1-i}  \right\| \nonumber\\
    &\le \left(\sum_{i=r+1}^{r+H+1} \kappa_B \kappa^3 (1-\gamma)^{i-r-1} \|w_{t-1-i} \|\right) + \|w_{t-1-r} \| \nonumber\\
    &\le \sum_{i=r}^{r+H+1} \kappa_B \kappa^3 (1-\gamma)^{i-r-1} \|w_{t-1-i} \| . \label{eq:dvdM}
\end{align}
First, we bound $ \sum_{t=0}^{\ft-1} \| \nabla_M f_t (M_t) \|_{\rm F}^2 $. By \eqref{eq:grad_pq_bound}--\eqref{eq:dvdM},
\begin{align}
   &\bbE\left[\left|\nabla_{M_{p,q}^{[r]}} f_t (M) \right|^2\right] \le G_c^2 \bbE\left[  \left( \|y_t (M)\|   \left\|\frac{\partial y_t(M)}{\partial M_{p,q}^{[r]}}\right\|  + \| v_t(M) \|  \left\| \frac{\partial v_t (M)}{\partial M_{p,q}^{[r]}} \right\| \right)^2 \right] \nonumber\\
    &\le G_c^2 \bbE\Biggl[  \Biggl\{ \left( \sum_{i=0}^{2H} (2H+1) \kappa_B^2 \kappa^5 (1-\gamma)^{i-1} \|w_{t-1-i}\| \right) \left(\sum_{i=r+1}^{r+H+1} \kappa_B \kappa^2 (1-\gamma)^{i-r-1} \|w_{t-1-i} \|\right) \nonumber\\
    &\quad + \left(\sum_{i=0}^{2H}  (2H+3)\kappa_B^2 \kappa^{6} (1-\gamma)^{i-1} \|w_{t-1-i}\| \right) \left( \sum_{i=r}^{r+H+1} \kappa_B \kappa^3 (1-\gamma)^{i-r-1} \|w_{t-1-i} \| \right) \Biggr\}^2  \Biggr] \nonumber\\
    &\le G_c^2 \bbE\Biggl[ \Biggl\{ 2\left(\sum_{i=0}^{2H}  (2H+3)\kappa_B^2 \kappa^{6} (1-\gamma)^{i-1} \|w_{t-1-i}\| \right) \left( \sum_{i=r}^{r+H+1} \kappa_B\kappa^3 (1-\gamma)^{i-r-1} \|w_{t-1-i}\| \right) \Biggr\}^2  \Biggr] \nonumber\\
    &\le 4 G_c^2 (2H+3)^2 \kappa_B^6 \kappa^{18} \nonumber\\
    &\quad \times \bbE\left[ \left(  \sum_{i=0}^{2H} (1-\gamma)^{i-1} \|w_{t-1-i}\| \right)^4   \right]^{1/2} \bbE\left[ \left( \sum_{i=r}^{r+H+1}  (1-\gamma)^{i-r-1} \|w_{t-1-i}\| \right)^4  \right]^{1/2} , \label{eq:nabla_square_bound}
\end{align}
where in the last line, we used H\"{o}lder's inequality.
For \eqref{eq:nabla_square_bound}, we have
\begin{align}
    \bbE\left[ \left(  \sum_{i=0}^{2H} (1-\gamma)^{i-1} \|w_{t-1-i}\| \right)^4   \right] &\le \bbE\left[ \left(  \max_{i\in \bbra{0,2H}} \|w_{t-1-i}\|\sum_{i=0}^{2H} (1-\gamma)^{i-1} \right)^4   \right] \nonumber\\
    &\le \frac{1}{\gamma^4 (1-\gamma)^4} \bbE\left[ \max_{i\in \bbra{0,2H}} \|w_{t-1-i}\|^4 \right] \nonumber\\
    &\le \frac{2H+1}{\gamma^4 (1-\gamma)^4} \max_{i\in \bbra{0,2H}} \bbE\left[ \|w_{t-1-i}\|^4\right] \nonumber\\
    &\le \frac{(2H+1)\sigma_w^4}{\gamma^4 (1-\gamma)^4}, \nonumber\\
    \bbE\left[ \left( \sum_{i=r}^{r+H+1}  (1-\gamma)^{i-r-1} \|w_{t-1-i}\| \right)^4  \right]&\le \frac{(H+2)\sigma_w^4}{\gamma^4 (1-\gamma)^4} . \nonumber
\end{align}
Therefore, we obtain
\begin{align}
    \bbE\left[\left|\nabla_{M_{p,q}^{[r]}} f_t (M) \right|^2\right] \le \frac{4 G_c^2 (2H+3)^3 \kappa_B^6 \kappa^{18}\sigma_w^4}{\gamma^4 (1-\gamma)^4} , \label{eq:nabla_square_2}
 \end{align}
which yields for any $ \thr_1 > 0 $,
\begin{align}
    \bbP\left( \sum_{t=0}^{\ft-1} \|\nabla_M f_t (M) \|_{\rm F}^2 \ge \thr_1  \right) &\le \frac{1}{\thr_1} \bbE\left[ \sum_{t=0}^{\ft-1} \|\nabla_M f_t (M) \|_{\rm F}^2 \right] \nonumber\\
    &\le \frac{1}{\thr_1} \bbE\left[ \sum_{t=0}^{\ft-1} \sum_{p=1}^{n_\rmu} \sum_{q=1}^{n_\rmx} \sum_{r=0}^{H-1} |\nabla_{M_{p,q}^{[r]}} f_t (M)|^2  \right] \nonumber\\
    &\le \frac{4\dime^2 G_c^2 (2H+3)^4 \kappa_B^6 \kappa^{18}\sigma_w^4 \ft}{\thr_1 \gamma^4 (1-\gamma)^4} . \nonumber
\end{align}
For any $ \ft \ge 3 $, $ \coe > 0 $, and $ M \in \calM $, by letting $ H = \lceil 2\gamma^{-1} \log \ft  \rceil $, $ \thr_1 = \coe \ft (\log \ft)^4 $ in the above inequality, we get
\begin{align}
    \bbP\left( \sum_{t=0}^{\ft-1} \|\nabla_M f_t (M) \|_{\rm F}^2 \ge \coe  \ft (\log \ft)^4  \right)
    &\le \frac{4\dime^2 G_c^2 (\frac{4}{\gamma}+\frac{5}{\log \ft})^4 \kappa_B^6 \kappa^{18}\sigma_w^4 }{\coe \gamma^4 (1-\gamma)^4} \nonumber\\
    &\le \frac{4 \cdot 9^4 \dime^2 G_c^2  \kappa_B^6 \kappa^{18}\sigma_w^4 }{\coe \gamma^8 (1-\gamma)^4}  .
\end{align}

Next, we bound $ \sum_{t=0}^{\ft-1} \sum_{i=1}^{\min\{H+1,t\}} \sum_{k=1}^i
\|
    \nabla f_{t-k}(M_{t-k})
\|_{\rm F} $.
By H\"{o}lder's inequality, \eqref{eq:nabla_square_bound}, and \eqref{eq:nabla_square_2}, we have
\begin{align}
    &\bbE\left[ \|y_{t-k} (M)\|  \left\|\frac{\partial y_{t-k}(M)}{\partial M_{p,q}^{[r]}}\right\|  + \| v_{t-k}(M) \|  \left\| \frac{\partial v_{t-k} (M)}{\partial M_{p,q}^{[r]}} \right\|\right] \nonumber\\
    &\le \bbE\left[ \left( \|y_{t-k} (M)\|  \left\|\frac{\partial y_{t-k}(M)}{\partial M_{p,q}^{[r]}}\right\|  + \| v_{t-k}(M) \|  \left\| \frac{\partial v_{t-k} (M)}{\partial M_{p,q}^{[r]}} \right\| \right)^2 \right]^{1/2} \nonumber\\
    &\le \frac{2 (2H+3)^{3/2} \kappa_B^3 \kappa^{9}\sigma_w^2 }{\gamma^2 (1-\gamma)^2} .
\end{align}
Hence, it follows from \eqref{eq:grad_pq_bound} that
\begin{align}
    &\bbE\left[\sum_{t=0}^{\ft-1} \sum_{i=1}^{\min\{H+1,t\}} \sum_{k=1}^i
\|
    \nabla f_{t-k}(M)
\|_{\rm F} \right] \nonumber\\
&\le\sum_{t=0}^{\ft-1} \sum_{i=1}^{\min\{H+1,t\}} \sum_{k=1}^i
\bbE\left[\left(\sum_{p=1}^{n_\rmu} \sum_{q=1}^{n_\rmx} \sum_{r=0}^{H-1} \left|\nabla_{M_{p,q}^{[r]}} f_t (M) \right|^2 \right)^{1/2} \right]\nonumber\\
&\le\sum_{t=0}^{\ft-1} \sum_{i=1}^{\min\{H+1,t\}} \sum_{k=1}^i
\sum_{p=1}^{n_\rmu} \sum_{q=1}^{n_\rmx} \sum_{r=0}^{H-1} \bbE\left[\left|\nabla_{M_{p,q}^{[r]}} f_t (M) \right| \right]\nonumber\\
&\le G_c\sum_{t=0}^{\ft-1} \sum_{i=1}^{\min\{H+1,t\}} \sum_{k=1}^i
\sum_{p=1}^{n_\rmu} \sum_{q=1}^{n_\rmx} \sum_{r=0}^{H-1} \bbE\left[ \|y_t (M)\|   \left\|\frac{\partial y_t(M)}{\partial M_{p,q}^{[r]}}\right\|  +  \| v_t(M) \|  \left\| \frac{\partial v_t (M)}{\partial M_{p,q}^{[r]}} \right\| \right]\nonumber\\
&\le G_c\sum_{t=0}^{\ft-1} \sum_{i=1}^{\min\{H+1,t\}} \sum_{k=1}^i
\sum_{p=1}^{n_\rmu} \sum_{q=1}^{n_\rmx} \sum_{r=0}^{H-1} \frac{2 (2H+3)^{3/2} \kappa_B^3 \kappa^{9}\sigma_w^2 }{\gamma^2 (1-\gamma)^2}  \nonumber\\
&\le \frac{2\dime^2 G_c (2H+3)^{9/2} \kappa_B^3 \kappa^{9}\sigma_w^2 T }{\gamma^2 (1-\gamma)^2} . \nonumber
\end{align}
Therefore, for any $ \thr_2 > 0 $, it holds that
\begin{align}
    \bbP\left( \sum_{t=0}^{\ft-1} \sum_{i=1}^{\min\{H+1,t\}} \sum_{k=1}^i
    \|
        \nabla f_{t-k}(M)
    \|_{\rm F} \ge \thr_2 \right) \le \frac{2 \dime^2 G_c (2H+3)^{9/2} \kappa_B^3 \kappa^{9}\sigma_w^2 T }{\thr_2 \gamma^2 (1-\gamma)^2} . \nonumber
\end{align}
For any $ \ft \ge 3 $ and $ \coe > 0 $, by setting $ H = \lceil 2\gamma^{-1} \log \ft  \rceil $, $ \thr_2 = \coe \ft (\log \ft)^{9/2} $, we obtain
\begin{align}
    \bbP\left( \sum_{t=0}^{\ft-1} \sum_{i=1}^{\min\{H+1,t\}} \sum_{k=1}^i
    \|
        \nabla f_{t-k}(M)
    \|_{\rm F} \ge \coe \ft (\log \ft)^{9/2} \right) &\le \frac{2\dime^2G_c (\frac{4}{\gamma}+\frac{5}{\log \ft})^{9/2} \kappa_B^3 \kappa^{9}\sigma_w^2 }{\coe  \gamma^2 (1-\gamma)^2} \nonumber\\
    &\le \frac{2\cdot 3^{9} \dime^2 G_c \kappa_B^3 \kappa^{9}\sigma_w^2 }{\coe  \gamma^{13/2} (1-\gamma)^2} ,
\end{align}
which completes the proof.

\section{Proof of Proposition~\ref{prop:hessian}}\label{app:smoothness}
For notational simplicity, we drop the superscript $ K $ of $ y_t^{K} $ and $ v_t^{K} $. First, we have
    \begin{align}
        \nabla_{M}^2 f_t (M) &=  \nabla_{M}^2 c_t \left(  y_t (M), v_t (M) \right) \nonumber\\
        &=  \begin{bmatrix}
            \calJ_{y_t} \\ \calJ_{v_t}
        \end{bmatrix}^\top
         \nabla_{x,u}^2 c_t (y_t (M), v_t (M)) 
        \begin{bmatrix}
            \calJ_{y_t} \\ \calJ_{v_t}
        \end{bmatrix} , \label{eq:hessian_f_c}
    \end{align}
    where $ \calJ_{y_t} $ and $ \calJ_{v_t} $ are the Jacobian matrices of $ y_t $ and $ v_t $ with respect to $ M $, respectively.
Hence, by Assumption~\ref{ass:strong_convex_cost},
\begin{align}
    \left\| \nabla_{M}^2 f_t (M) \right\|_{\rm F}
    &\le  \left\| \begin{bmatrix}
        \calJ_{y_t} \\ \calJ_{v_t}
    \end{bmatrix}\right\|_{\rm F}^2
     \|\nabla_{x,u}^2 c_t (y_t (M), v_t (M)) \|_{\rm F}
      \le \beta \left\| \begin{bmatrix}
        \calJ_{y_t} \\ \calJ_{v_t}
    \end{bmatrix}\right\|_{\rm F}^2 . \label{eq:hessian}
\end{align}
Here, we have
\begin{align}
    \left\|\begin{bmatrix}
        \calJ_{y_t} \\ \calJ_{v_t}
    \end{bmatrix}\right\|_{\rm F}^2  = \sum_{p=1}^{\dime_\rmu} \sum_{q=1}^{\dime_\rmx} \sum_{r=0}^{H-1} \left( \left\| \frac{\partial y_t(M)}{\partial M_{p,q}^{[r]}}  \right\|^2 + \left\| \frac{\partial v_t(M)}{\partial M_{p,q}^{[r]}}  \right\|^2 \right) . \label{eq:jacobi_bound}
\end{align}

We will bound the right-hand side of \eqref{eq:jacobi_bound} based on \eqref{eq:dydM} and \eqref{eq:dvdM}.
For any $ \thr_1 > 0 $, it holds that
\begin{align}
    &\bbP\left(\max_{t\in \bbra{0,\ft-1}} \|w_t\| \ge \thr_1 \right) \nonumber\\
    &\le \bbP\left( \max_{t\in \bbra{0,\ft-1}, i \in \bbra{1,\dime_\rmx}} \dime_\rmx |w_{t,i}| \ge \thr_1 \right) \nonumber\\
    &\le \bbP\left( \max_{t\in \bbra{0,\ft-1}, i \in \bbra{1,\dime_\rmx}} w_{t,i} \ge \frac{\thr_1}{\dime_\rmx}  \right) + \bbP\left( \max_{t\in \bbra{0,\ft-1}, i \in \bbra{1,\dime_\rmx}} (-w_{t,i}) \ge \frac{\thr_1}{\dime_\rmx} \right) \nonumber\\
    &\le \frac{\dime_\rmx}{\thr_1} \bbE\left[ \max_{t\in \bbra{0,\ft-1}, i \in \bbra{1,\dime_\rmx}} w_{t,i}    \right] + \frac{\dime_\rmx}{\thr_1} \bbE\left[ \max_{t\in \bbra{0,\ft-1}, i \in \bbra{1,\dime_\rmx}} (-w_{t,i})    \right] . \label{eq:wnorm_bound}
\end{align}
Moreover, by Jensen's inequality, for any $ \lambda > 0 $,
\begin{align}
    \exp\left( \lambda \bbE\left[\max_{t\in\bbra{0,\ft-1},i\in \bbra{1,\dime_\rmx}} w_{t,i} \right]  \right) &\le \bbE\left[ \exp\left( \lambda \max_{t,i} w_{t,i} \right)  \right] \nonumber\\
    &\le \bbE\left[ \sum_{t=0}^{\ft-1} \sum_{i=1}^{\dime_\rmx} \exp\left( \lambda w_{t,i}  \right) \right] \nonumber\\
    &\le \ft \dime_\rmx \exp\left( \frac{\lambda^2 \sigma_w^2}{2\dime_\rmx} \right), \nonumber
\end{align}
where we used the sub-Gaussianity~\eqref{eq:subgauss_def}.
The above inequality yields that for any $ \ft > \dime_\rmx^{-1} $,
\begin{align}
    \bbE\left[\max_{t\in\bbra{0,\ft-1},i\in \bbra{1,\dime_\rmx}} w_{t,i} \right] &\le \frac{\log (\ft \dime_\rmx)}{\lambda}  + \frac{\lambda \sigma_w^2}{2\dime_\rmx} \le \sigma_w\sqrt{\frac{2 \log (\ft \dime_\rmx)}{\dime_\rmx}} . \label{eq:wmax1}
\end{align}
Similarly, for any $ \ft > \dime_\rmx^{-1} $, we have
\begin{align}
    \bbE\left[\max_{t\in\bbra{0,\ft-1},i\in \bbra{1,\dime_\rmx}} (-w_{t,i}) \right] \le\sigma_w\sqrt{\frac{2 \log (\ft \dime_\rmx)}{\dime_\rmx}} . \label{eq:wmax2}
\end{align}
By \eqref{eq:wnorm_bound}--\eqref{eq:wmax2}, we obtain
\begin{align}
    \bbP\left(\max_{t\in \bbra{0,\ft-1}} \|w_t\| \ge d_1 \right) 
    \le \frac{2  \sigma_w\sqrt{2 \dime_\rmx \log (\ft \dime_\rmx)} }{\thr_1} , ~~ \forall \thr_1 > 0, \ \ft > \dime_\rmx^{-1} .
\end{align}

Hence, it follows from \eqref{eq:dydM} and \eqref{eq:dvdM} that for any $ \thr_1 > 0 $, $ \ft > \dime_\rmx^{-1} $, and $  M \in \calM $,
\begin{align}
    &\bbP\Biggl( \max_{t\in \bbra{0,\ft-1}, p\in \bbra{1,\dime_u}, q \in \bbra{1,\dime_\rmx}, r\in \bbra{0,H-1}} \left\| \frac{\partial y_t(M)}{\partial M_{p,q}^{[r]}}  \right\| \le \frac{\kappa_B \kappa^2 \thr_1}{\gamma}, ~~ \max_{t,p,q,r} \left\| \frac{\partial v_t(M)}{\partial M_{p,q}^{[r]}}  \right\| \le \frac{\kappa_B \kappa^3 \thr_1}{\gamma(1-\gamma)} \Biggr) \nonumber\\
    &\ge  1 - \frac{2\sigma_w\sqrt{2 \dime_\rmx\log (\ft \dime_\rmx)} }{\thr_1} . \nonumber
\end{align}
By combining this with \eqref{eq:hessian} and \eqref{eq:jacobi_bound}, we get for any $ \thr_1 > 0 $, $ \ft > \dime_\rmx^{-1} $, and $ M\in \calM $,
\begin{align}
    \bbP\left(\max_{t\in \bbra{0,\ft-1}}\left\| \nabla_{{M}}^2 f_t (M) \right\|_{\rm F} \le \beta \dime^2 H \left(\frac{\kappa_B \kappa^2 \thr_1}{\gamma} + \frac{\kappa_B \kappa^3 \thr_1}{\gamma(1-\gamma)} \right) \right)
    \ge 1 - \frac{2\sigma_w\sqrt{2 \dime_\rmx\log (\ft \dime_\rmx)} }{\thr_1} .
\end{align}
For any $ \ft \ge 3 $ and $ \coe > 0 $, by letting $ H = \lceil 2\gamma^{-1} \log \ft  \rceil $, $ \thr_1 = \coe \sqrt{\log \ft} $ in the above, we obtain
\begin{align}
    &\bbP\left(\max_{t\in \bbra{0,\ft-1}}\left\| \nabla_{M}^2 f_t (M) \right\|_{\rm F} \le  \frac{2\kappa_B \kappa^3 \beta \dime^2 (2\gamma^{-1} \log \ft + 1) \coe \sqrt{\log \ft}}{\gamma(1-\gamma)} \right) \nonumber\\
    &\ge  1 -\frac{2\sigma_w\sqrt{2\dime_\rmx (1+ \log \dime_\rmx)  }}{\coe} . \nonumber
\end{align}
The desired result is a straightforward consequence of the above inequality.

\section{Proof of Theorem~\ref{thm:log_regret}}\label{app:log_regret}
Since the bound \eqref{eq:regret_oco_log} does not cover the time interval $ \bbra{0,H+2} $, we first provide a bound for $ \bbra{0,H+2} $.
For any $ M \in \calM $, we have
\begin{align}
    | \surr_t (M_{t-1-H:t}) - f_t(M) | \le | \surr_t (M_{t-1-H:t}) - f_t(M_t) | + | f_t (M_t) - f_t(M) | . \nonumber
\end{align}
Assume that $ L_\rmc $ and $ L_\rmf $ satisfy \eqref{eq:lipschitz} and \eqref{eq:lipschitz2}.
Then, for any $ M_t, M\in \calM $, it holds that
\begin{align}
    | \surr_t (M_{t-1-H:t}) - f_t(M_t) | &\le L_\rmc \sum_{i=1}^{H+1} \|M_t - M_{t-i} \|_{\rm F} \nonumber\\
    &\le L_\rmc \sum_{i=1}^{H+1} \sum_{k=1}^i \|M_{t-k+1} - M_{t-k} \|_{\rm F}  \nonumber\\
    &\le L_\rmc D (H+1)^2, \nonumber\\
    |f_t (M_t) - f_t(M)| &\le L_{\rm f} D ,   \nonumber
\end{align}
where we used $ \sup_{M,M'\in \calM} \|M - M' \|_{\rm F} \le D $.
Hence, we obtain
\begin{align}
    \sum_{t=0}^{H+2} | \surr_t (M_{t-1-H:t}) - f_t(M) | \le L_\rmc (H+3)^3 D + L_\rmf (H+3)^2 D . \nonumber
\end{align}

Under the assumptions of Proposition~\ref{prop:log_regret}, for any $ \delta' \in (0,1] $, it holds that
\begin{align}
    &\bbP\Biggl(\sum_{t = 0}^{\ft-1} \surr_t (M_{t-1-H:t}) -  \sum_{t = 0}^{\ft-1} f_t (M) \nonumber\\
    &\quad \lesssim \alpha' D^2 H_+ + \frac{L_{\rm f}^2 H_+}{\alpha'} \log \left( \frac{1 + \log (\Ee + \alpha' D^2)}{\delta'} \right) + \frac{ (2L_{\rmc} + \dime^2 L_\rmf) H_+ + \beta'}{\alpha'} L_\rmf H_+ \log \ft \nonumber\\
    &\quad  + L_\rmc D(H_+ +1)^3 + L_\rmf D(H_+ +1)^2 \Biggr) \ge 1 - \delta', \nonumber
\end{align}
where $ H_+ := H + 2 $.
Therefore, by Lemmas~\ref{lem:lipshitz}, \ref{lem:strong_convex}, and \ref{lem:lipschitz_global} and Proposition~\ref{prop:hessian}, for $ \bar{L}_\coe := \frac{4G_c \coe^2 (\log \ft)^{5/2}}{\sqrt{\gamma}} $ and $ \bar{\beta}_\coe := \frac{6\kappa_B \kappa^3 \beta \dime^2 \coe (\log \ft)^{3/2}}{\gamma^2 (1-\gamma)} $, it holds that for any $ \coe > 0 $ and $ \ft \ge 3 $,
\begin{align}
    &\bbP\Biggl(\sum_{t = 0}^{\ft-1} \surr_t (M_{t-1-H:t}) -  \sum_{t = 0}^{\ft-1} f_t (M)  \lesssim \wtilde{\alpha}  D^2 H_+ + \frac{ \bar{L}_{C}^2 H_+}{\wtilde{\alpha}} \log \left( \frac{1 + \log (\Ee + \wtilde{\alpha} D^2)}{\delta'} \right) \nonumber\\
    &\qquad + \frac{ (2+\dime^2)  \bar{L}_{C}H_+ +  \bar{\beta}_\coe }{\wtilde{\alpha}} \bar{L}_{C} H_+ \log \ft  + 2  D \bar{L}_{C} (H_+ +1)^3 \Biggr) \nonumber\\
    &\ge 1 - \delta' - \frac{71\sigma_w \kappa_B^2 \kappa^6}{C\gamma^2 (1-\gamma)} - \frac{2\sigma_w \sqrt{2\dime_\rmx (1+\log \dime_\rmx)}}{\coe} . \label{eq:third_log}
\end{align}

By the proof of Theorem~\ref{thm:regret}, the first and second lines \eqref{eq:regret_decomposition1} and \eqref{eq:regret_decomposition2} of the decomposition of the regret with $ H = \lceil 2\gamma^{-1} \log \ft  \rceil $ can be bounded as follows:
\begin{align}
    &\bbP\Biggl( \min_{M\in \calM} \sum_{t=0}^{T-1} f_t (M) -  \sum_{t=0}^{T-1} c_t (x_t^{K^*},u_t^{K^*}) \nonumber\\
    &\qquad + \sum_{t=0}^{T-1} c_t (x_t^K (M_{0:t-1}), u_t^K(M_{0:t}))  - \sum_{t=0}^{T-1} \surr_t (M_{t-1-H:t}) \le 6 G_c C^2 (\log \ft)^{2} \Biggr) \nonumber\\
    &\ge 1 - \frac{130 \sigma_w \kappa_B^2 \kappa^8}{C \gamma^2 (1-\gamma)} , ~~ \forall \coe > 0 . \label{eq:first_second}
\end{align}
In summary, by \eqref{eq:regret_decomposition} with \eqref{eq:third_log} and \eqref{eq:first_second}, we arrive at
\begin{align}
    &\bbP\Biggl( \regret_\ft^{\rm diag} \lesssim \wtilde{\alpha} D^2 H_\ft + \frac{\bar{L}_{\coe}^2 H_\ft}{\wtilde{\alpha}} \log \left( \frac{1 + \log (\Ee + \wtilde{\alpha} D^2)}{\delta'} \right) \nonumber\\
    &\qquad + \frac{ (2+\dime^2)  \bar{L}_{\coe} H_\ft   + \bar{\beta}_\coe }{\wtilde{\alpha}} \bar{L}_\coe H_\ft \log \ft   + 2  D \bar{L}_\coe (H_\ft+1)^3 + 6 G_c C^2 (\log \ft)^{2}  \Biggr) \nonumber\\
    &\ge 1 - \delta' - \frac{201 \sigma_w \kappa_B^2 \kappa^8}{C \gamma^2 (1-\gamma)}  - \frac{2\sigma_w\sqrt{2\dime_\rmx (1+ \log \dime_\rmx)  }}{\coe},  ~~ \forall \coe  > 0, \ \ft \ge 3 . \nonumber
\end{align}
Lastly, by letting $ \delta' = \delta / 2 $, $ \frac{201 \sigma_w \kappa_B^2 \kappa^8}{C \gamma^2 (1-\gamma)}  + \frac{2\sigma_w\sqrt{2\dime_\rmx (1+ \log \dime_\rmx)  }}{\coe} = \delta/2 $, we obtain the desired result.







\end{appendices}

\bibliography{acml2025_online_control}



\end{document}